\documentclass[a4paper]{amsart}
\usepackage{a4}
\usepackage{amsmath}
\usepackage{amssymb}
\usepackage[applemac]{inputenc}
\usepackage{graphicx}
\usepackage{color}
\newcommand{\PP}{ \mathbb{P}}
\newcommand{\FF}{ \mathbb{F}}
\newcommand{\QQ}{ \mathbb{Q}}

 \newcommand{\Z}{{\mathbb Z}}

\newcommand{\EE}{{\mathbb E}}

\newtheorem{remark}{\textbf{Remark}}[section]

\newtheorem{lemma}{\textbf{Lemma}}[section]
\newtheorem{theorem}{\textbf{Theorem}}[section]

\newtheorem{proposition}{\textbf{Proposition}}[section]

\newtheorem{definition}{\textbf{Definition}}[section]

\numberwithin{equation}{section}

\title{ Sensitivity analysis for expected utility maximization in incomplete Brownian market models }

\author{
Julio Backhoff Veraguas  \and     Francisco J. Silva   }
\def\dd{{\rm d}}

\def\weight(#1,#2){c_{#1,#2}}
 \if {
  
 } \fi
\if {

}{{\caption}}
} \fi

\def\E{\mathcal{E}}
\def\F{\mathcal{F}}

\def\P{\mathcal{P}}

\def\SS{\mathcal{S}}

\def\V{\mathcal{V}}

\def\X{\mathcal{X}}
\def\Y{\mathcal{Y}}
\def\Z{\mathcal{Z}}
\DeclareMathOperator*{\esssup}{ess\,sup}

\def\1B{{\bf  1}}

\newcommand{\RR}{\mathbb{R}}

\newcommand\be{\begin{equation}}
\newcommand\ee{\end{equation}}
\newcommand\ba{\begin{array}}
\newcommand\ea{\end{array}}
\newcommand{\bean}{\begin{eqnarray*}}
\newcommand{\eean}{\end{eqnarray*}}

\thanks{  The first author is most grateful for partial support by the Austrian Science Fund (FWF) under grant Y782-N25 and the European Research Council (ERC) under grant FA506041, as well as to Humboldt-Universit\"at zu Berlin and the funding by the Berlin Mathematical School. Institute of Statistics and Mathematic methods in Economics , Vienna University of Technology (julio.backhoff@tuwien.ac.at).}

\thanks{   The second author acknowledges partial support by the Gaspar Monge Program for Optimization and Operation Research (PGMO). Institut de recherche XLIM-DMI, UMR-CNRS 7252 Facult\'e des sciences et techniques 
Universit\'e de Limoges, 87060 Limoges, France (francisco.silva@unilim.fr) }

\begin{document}

\maketitle
\begin{abstract}  
We examine the issue of sensitivity with respect to model parameters for the problem of utility maximization from final wealth in an incomplete Samuelson model and mainly, but not exclusively, for utility functions of  positive-power type. The method consists in moving the parameters through change of measure, which we call a \textit{weak perturbation},  decoupling the usual wealth equation from the varying parameters. By rewriting the maximization problem in terms of a convex-analytical support function of a weakly-compact set, crucially leveraging on the work \cite{BFrobustez1}, the previous formulation let us prove the Hadamard directional differentiability of the value function w.r.t.\ the drift and interest rate parameters, as well as for volatility matrices under a stability condition on their Kernel, and derive explicit expressions for the directional derivatives. We contrast our proposed weak perturbations against what we call \textit{strong perturbations}, where the wealth equation is directly influenced by the changing parameters.  Contrary to conventional wisdom, we find that both points of view generally yield different sensitivities unless e.g.\ if initial parameters and their perturbations are deterministic.
\end{abstract}

%\begin{keywords} 

%\begin{AMS}\end{AMS}

\section{Introduction}

The problem of continuous-time utility maximization in financial market models has a long and rich history going back to Merton in \cite{Merton}-\cite{MertonJET}, himself inspired in the work of Mirrlees and Samuelson in discrete times. {The research on this topic continued in the eighties  through the works of Pliska \cite{PliskaOpt},  Karatzas et al. (see e.g.\ \cite{KaraLehSchr,KaraLehSchr91}), Cox and Huang \cite{CoxHuang}  and then probably culminated in the nineties with the general treatment of Kramkov and Schachermayer in \cite{KrSch}}. Naturally a comprehensive list would have to cover the works of many other people, but we do not intend to be exhaustive here and instead convey the interested reader to the books \cite{KaraShreveFinance} and \cite{Pham} for details. What all these works have in common, is that they provide an insight into the decision making problem of how to best select a portfolio from a given continuous-time, stochastic market model under the optimality criterion provided by the expected utility paradigm of von Neumann-Morgenstern.

It goes without a saying that in modelling the decision-making in such way, several parameters have to be chosen and therefore both the optimal portfolio rule and the optimal expected utility derived from it will be a function of these. Yet only recently the behaviour of the expected utility maximization problem in terms of its parameter-dependence has gained attention. In \cite{KrSch}, for the case of general semimartingale models and an agent optimizing expected utility from final wealth only and no random endowment, the first-order sensitivity of the problem's value function (i.e.\ the optimal value) with respect to the initial wealth of the agent is studied, extending earlier results in \cite{PliskaOpt}. More recently and in a similar setting, a second-order analysis of the value function is performed in \cite{Kramkov} and even the first-order sensitivity of the optimizing wealth is carried out. A different trait in the literature has been the study of the stability (i.e.\ continuity) of the value function with respect to  the so-called market price of risk or Sharpe ratio, which is a dynamic and stochastic parameter, heuristically measuring how much a given price model is away of a risk-neutral one (given by its martingale component). This analysis was performed in \cite{ZitLar} initially  (see also e.g.\ \cite{MWstability} for recent developments), and then extended in \cite{ZitKar} for the case when a random endowment is present. The last article goes beyond that and actually proves stability of utility-based prices and admits misspecification of the utility functions themselves (see also  \cite{Larsenpref} and the references given, for more on this subject).  The previous articles focus on equivalent perturbations of a reference probability measure or a reference price process; recently \cite{Weston} has showed that for non-equivalent perturbations the problem may be unstable/discontinuous. 

%More substantially than just differing in the type of exponent-regime (which is something that can certainly be amended both in \cite{prepLarsen} and here; see e.g.\ Remark \ref{Ugral}), our seemingly simpler market model cannot be generally put in the form of a semimartingale model without portfolio constraints (as assumed in \cite{prepLarsen}), and more importantly, we propose and develop a different methodology to obtain and compute the desired sensitivities, as we now explain. 

{In this article we focus on the first-order sensitivity analysis of the optimal value  of the expected utility maximization problem with respect to the {\it market price of risk} and the {\it drift and volatility} coefficients of the model. We work in the classical setting where   the utility function  is defined  on the positive half line, in the absence of consumption and random endowments,
% coefficients in a direct way, this is, without resorting to the dual problem. 
and we restrict ourselves to a Brownian filtration and the so-called  Samuelson price model (e.g.\ geometric Brownian motion),  {which can be incomplete}. In this framework, it is to be expected from the general stochastic maximum principle of \cite{CanKara95} (specifically {Section} 2 therein) and recent results in \cite{BSsensibilidad1}, that the desired differentiability can be computed with the help of the adjoint states appearing in the stochastic maximum principle. 
%% reduced to a question of stability of a coupled system of Forward-Backward Stochastic Differential equations (the Pontryagin's FSBDE), whose solution also yields the desired sensitivities in an explicit manner. 
 There are however several delicate points for this roadmap to work, the main one being that market prices of risk are multiplied by the decision variable (portfolio weigths) in the controlled wealth equation, and so standard convex analysis arguments for convex perturbations are not applicable. Alternative arguments based in abstract optimization theory (see \cite[Chapter 4]{BonSha} for the general theory and  \cite{BSsensibilidad1} for its application to stochastic control) seem diffucult to apply since they require  a normed vector space setting which is a priori absent in our problem. As a matter of fact,   decision variables are a priori only almost surely square integrable with respect to the time variable. For these reasons, we choose in this article a different approach still allowing for a direct treatment of the first-order sensitivity question.}

%{\color{green}To do this, we reformulate what parameter uncertainty/misspecification is, to start with. This is done along the same philosophy as it has been done in robust or worst-case stochastic optimization for years now, namely, encoding } 

 Let us be precise as to how we interpret parameter uncertainty/misspecification in this article. We take the widespread point of view of robust or worst-case stochastic optimization, in which one encodes uncertain parameters in uncertain probability measures under which the stochastic optimization problems are to be defined. See e.g.\ \cite{Quenez} or \cite{EpsteinChen} in the context of model-misspecification and Knightean uncertainty in economics, or \cite{ZitKar} for the question of stability in utility maximization / utility-based prices. Accordingly, we postulate that full knowledge of the parameters of a problem amounts to, in our case, a complete description of the controlled wealth equation, meaning concretely the drift, interest rate and volatility coefficients (and hence the market price of risk). Parameter uncertainty means for us that the actual possible trajectories of the controlled system may have different ``probability weights'' than {those} specified by the law on {the} path space induced by the controlled equation under the exact, ``real'' parameters. Consequently, the expected utility maximization problem under a perturbation of a ``real'' parameter consists for us in perturbing the reference probability measure away from the law induced by the ``real'' controlled equation (that is, the one given the ``real'' parameters) yet otherwise leaving such ``real'' controlled equation fixed in the process. {Naturally, the perturbation of the probability measure is defined with the help of Girsanov's theorem and the optimal value of the new problem is   referred to as the
%
% Our definition of we  This is evidently done with the aid of a Girsanov transform, after taking advantage of the usual trick of reducing the parameter space to market prices of risk, which is possible under our assumptions.  
{\it weakly perturbed} value function.} In this work, we shall study the differentiability and compute the directional derivatives of this weakly perturbed value function with respect to the drift and volatility coefficients, computed in a neighbourhood of the ``real'' parameters.  As for all the articles around the topic of stability and sensitivity of the expected utility maximization problem already cited, only \cite{ZitKar} takes this point of view. The others consider {{\it strong}} perturbations of the problem, {meaning that} the reference probability measure is kept fixed and the equations are perturbed. 

%{\color{green}The approach in the previous paragraph, {through the adjoint states appearing in Pontryagin's maximum principle and the techniques in \cite{BSsensibilidad1}}, would also be circumscribed to such a point of view.}

We remark that concurrently and independently from us, a related question has been posed and analyzed in \cite{prepLarsen} in the context of power utility functions with negative exponents and a semimartingale market model. We shall, on the contrary, focus our analysis on power utility functions with positive exponents, {and more generally on utility functions dominated from above by such positive-power functions (see Theorem \ref{teoremarkmenossimpli} and comments thereafter)}. In \cite{prepLarsen} the authors essentially study the dual problem and its associated dual value function, and from this they obtain the desired sensitivities of the primal, original problem. More substantially, % than just differing in the type of exponent-regime (which is something that could certainly be amended both in \cite{prepLarsen} and here), 
 the main difference with respect to our work is that in the cited article the sensitivity is studied in the strong sense (see our discussion after introducing  Assumption {\bf (H1)} in Section \ref{preliminares}) with respect to the market price of risk parameter. In our work, we emphasize the analysis of perturbations in the weak sense performed directly on the drift and the volatility terms.
%  As we will see, at this stage, this forces us to consider a restricted space for perturbations of the volatility parameter, namely, those which preserve the Kernel (see Remark \ref{kernelstability}). 
 { One of the advantages of the Brownian market model we consider is that it allows us to bring to light  some delicate issues relating market incompleteness and the type of perturbations we are able to handle. Indeed, in the weak formulation  we are forced to   consider a restricted space for perturbations of the volatility parameter, namely those which preserve the Kernel (see Remark \ref{kernelstability}). We believe that this  discussion is essential and it seems absent in the literature.}

%We point out that an abstract analysis of this problems is mathematically interesting only in the case of incomplete markets. As a matter of fact, in the case of complete markets, the optimal value function can be computed explicetly in terms of the the of the market price of risk and thus the sensitivities can be computed easily from the explicit form of the value function. 

{An additional} nice feature of our Brownian framework is that it allows us to compare the sensitivity analysis in the strong and weak senses {in a most transparent way}.  A detailed discussion about the differences between these approaches is provided in  Sections \ref{preliminares} and \ref{onvaluefunctions}. As we will see, the sensitivities of the value function obtained from strong or weak perturbations need not coincide, and we provide examples in Section \ref{subsec counter} for this situation.   This is at odds with the implicit conventional wisdom that ``it makes no difference how one perturbs parameters''. We shall also show in an example that the weak sensitivity can behave in a counterintuitive fashion. Both phenomena occur when the nominal parameters are non-deterministic, so the lesson is that one should be cautious when applying weak (i.e. Girsanov-type) perturbations in such a situation. Although we do not provide a sensitivity analysis for strong perturbations, we can guess how the associated sensitivities would look like (consistently with \cite{prepLarsen}), and compare them to our weak sensitivities. Using Bismut's integration by parts formula we find out exactly how these differ; see equation \eqref{eq heuristic} in Section \ref{onvaluefunctions}. It is also worth noticing that if both the  nominal and the perturbed market parameters are deterministic functions, the directional sensitivities do coincide under our hypotheses, as we show in Proposition \ref{igualdadfuerteydebil}. 

When performing the differentiability analysis of the weakly perturbed problem, we greatly rely on recent results having their origin in \cite{BFrobustez1} and \cite{JBtesis}. Indeed, the crucial fact is that  we may interpret the expected utility maximization problem as the computation of a convex-analytical support function of a weakly-compact convex set in an explicit Banach space. The usefulness of working with weak perturbations and the weakly perturbed value function is that its differentiability and directional derivatives can then be computed by adapting Danskin's Theorem for support functions and using the chain rule for directional derivatives. For this, the Fr\'echet directional differentiability of the Girsanov transform as an operator between essentially bounded integrands and elements in the pre-dual of the aforementioned Banach space has to be established. This issue poses most of the challenges in the present article.  { Our choice of dealing directly with the primal problem, via this support-function interpretation, is a second major distinction from \cite{prepLarsen}. }%We think it is good, and important, that any research question is assessed from as many angles as possible} %It seems natural to us, and important, that the sensitivity analysis is done from as many points of view as possible. Our article focuses on the  

In a nutshell our work has two original  contributions. The first one is to provide new sensitivity results for weakly perturbed problems and {fairly} precise expressions for the directional derivatives. The main tool here is, as discussed in the previous paragraph, a hidden compactness property of the feasible set in a natural topological space. {In fact, we consider this purely primal analysis as a methodological contribution of its own, as opposed to more classical points of view in mathematical finance such as duality or stochastic control.} The second contribution is the detailed discussion on the type of perturbations allowed as well as on the difference between weak and strong perturbations and their associated sensitivities. {Let us stress again that the simplicity of the market model we consider allows us to address the subtleties of the problem, and obtain the aforementioned contributions, in a clean and precise manner.}

The paper is structured as follows. In Section \ref{preliminares} we present our Samuelson model, define the  {strong/weak} perturbations and {strongly/weakly} perturbed value functions and describe our main result regarding differentiability of the value function {under weak perturbations}; Theorem \ref{teoremarkmenossimpli}. {Of equal importance, we also prove that in the case of deterministic parameters and perturbations the strongly and weakly perturbed value functions do coincide, whereas we also provide two simple examples showing that in the general case the strong and weak sensitivities can differ.} In Section \ref{utilmax} we provide for convenience of the reader a summary of the results in \cite{BFrobustez1} needed for our proofs. Section \ref{posssens} is the backbone of the article, where we prove the main sensitivity result.
{Then in Section \ref{onvaluefunctions} we present a discussion on how the strong and weak sensitivities are connected.}
% the second contribution of this work, consisting in a detailed analysis of the relationship between strongly and weakly perturbed value functions and providing two simple examples showing that the associated sensitivities in general differ.
Finally, in the appendix, we briefly study support functions and prove a needed adaptation of the classical Danskin's Theorem.

\section{Problem statement}
\label{preliminares}
We first fix some notations. In the entire article $\RR_{+}$ ($\RR_{++}$ respectively) will denote the set of non-negative (respectively strictly positive) real numbers. Given $T \in \RR_{++}$,  we consider a fixed filtered probability space $(\Omega,\F_T,\FF=\{\F_t\}_{t\leq T},\PP)$, where the filtration $\FF$ satisfies the usual assumptions (see e.g. \cite{Protter-book}). {Actually except for the results presented in Section 3, in which we survey some of the findings in \cite{BFrobustez1}, we will assume that $\FF$ is the   completed filtration of the Brownian motion defined therein}. We will denote  by $L^0$ (resp.\ $L^0_+$) the set of all $\F_T$-measurable functions (resp.\ non-negative ones), and by $L^{\infty,\infty}_{\FF}$ the set of essentially bounded real-valued progressively measurable processes endowed with the norm $\|\cdot\|_{\infty,\infty}$ defined as the least essential upper bound. Integration with respect to a measure $\QQ$ shall be denoted $\EE^{\QQ}$ except for $\QQ=\PP$, for which we reserve the notation $\EE$. Given a local continuous  martingale $M: \Omega \times [0,T] \to {\RR}$, we denote by 
$L^2_{loc}(M)$ the set of all  progressively measurable processes $H: \Omega \times [0,T] \to  \RR$ such that {$\PP( \int_{0}^{T} H^2_s \dd \langle M \rangle_s< +\infty)=1$,
%$L^2_{loc}(M)$ the set of all progressively measurable processes $H: \Omega \times [0,T] \to \RR^m$ such that {\color{blue}$\PP( \sum_{i=1}^m\int_{0}^{T} [H^i_s]^{\top} \dd \langle M^i \rangle_s H^i_s< +\infty)=1$,
 where $\langle M\rangle_{(\cdot)}$ denotes the quadratic variation process associated to $M$}. Finally, given a continuous semimartingale $Y$, we denote by $\E(Y)$, the {Dol\'eans-Dade} stochastic exponential, defined as the solution of $ Z_t= 1+ \int_{0}^{t}Z_s \dd Y_s$, for $t\in [0,T]$.

Let us consider a general Samuelson's price model for this section, where discounted prices evolve continuously as geometric Brownian motions with  progressively measurable drift and volatility coefficients. Specifically, suppose that the market consists of $d$ assets  $S^{1}, \hdots, S^{d}$ whose prices (denoted likewise) evolve under $\PP$ as
\be
\ba{rcl}
%\dd S^{0}(t)&=& r S^{0}(t) \dd t, \; \; \mbox{for } t \in [0,T], \; \; S^{0}(0)=1, \\[4pt]
\dd S_t&=&\mbox{diag}(S_t)\bar{\mu}_t\dd t+\mbox{diag}(S_t)\bar{\sigma}_t \dd W_t \; \; \mbox{for } \; t \in [0,T], \\[4pt]
S_0&=& {s_{0}} \in \RR^{d},\\[4pt]
\ea\ee
where  $S:=(S^{1}, \hdots, S^{d})$ and $W$ is a $\PP$-Brownian motion in $\RR^n$ ($n\geq d$). The precise properties on the processes $\bar{\mu} \in (L^{\infty, \infty}_{\FF})^{d}$ and $\bar{\sigma} \in (L^{\infty, \infty}_{\FF})^{d\times n}$ shall be given shortly and  will imply that  the financial market  is viable  and moreover standard (see e.g.\ \cite[Chapter 1]{KaraShreveFinance} or \cite[Chapter 7.2.4]{Pham} for these concepts and the modelling details). 

Given an initial wealth $x\in {\RR_{++}}$ and a  {\it self-financing portfolio} $\pi$ measured in units of wealth such that $\pi^i \in {L^{2}_{loc}(W^k)}$ ($i\in {1,\dots,d}$ and $k=1,\hdots, n$), which we denote $\pi\in\Pi$,
%{\color{blue}such that $\pi^i \in {L^{2}_{loc}(W^k)}$ for each $i\in {1,\dots,d}$ and $k\in {1,\dots,n}$, which we denote $\pi\in\Pi$},
 the associated {\it wealth process} $X$ is defined through the equation
%\be\label{vdebil}\ba{rcl}
%\dd X_t^{\pi}&=& \left[ r(t)X_t^{\pi} +\pi_t^{\top}(\mu_t-r_t{\bf 1}) \right] \dd t+\pi_t^{\top} \sigma_t \dd W_t, \hspace{0.4cm} \; t \in [0,T],\\[4pt]
%X_0^{\pi}&=&x. \ea
%\ee
\be\label{vdebil}\ba{rcl}
\dd X_t^{\pi}&=&  {\pi_t^{\top}\bar{\mu}_t} \dd t+\pi_t^{\top}\bar{ \sigma}_t \dd W_t \hspace{0.4cm} \mbox{for $ t\in [0,T]$},\\[4pt]
X_0^{\pi}&=&x. \ea
\ee
{In this work,} we consider the following utility maximization problem 
\be\label{problemaformulacionfuerte}
u({\bar{\mu},\bar{\sigma}}):= \sup  \left\{ \EE\left( U(X_{T}^{ \pi})\right) \; ; \;  \pi \in \Pi \; \; \mbox{and } \; X^{\pi}_t \geq 0 \; \;  \forall \; t\in [0,T],  \; \mbox{$\PP$-a.s.}\right\},
\ee
%where 
%$$ { \K:=\left\{{\color{blue} \pi\in\Pi}\; ; \;  {X^{\pi}_t} \geq 0, \; \; \forall \; t\in [0,T]%\right\}},$$
where $U:= \RR \to \RR \cup \{-\infty\}$ is a concave  {\it utility function}, whose properties will be specified in {Section \ref{utilmax}}, but for the time being we suppose that  $U(x)=-\infty$ if $x< 0$ and the restriction of $U$ to $\RR_{+}$ takes values in $\RR_{+}$ and is invertible. Since the financial market is viable, almost sure non-negativity of $X_{T}^{\pi}$ implies that $X_{t}^{\pi}\geq 0$ for all $t\in [0,T]$, $\PP$-a.s. Thus, 
$$u({\bar{\mu},\bar{\sigma}})= \sup_{  \pi\in\Pi} \EE\left( U(X_{T}^{ \pi})\right).$$   
{If we want to perform a sensitivity analysis with respect to the new parameters} $\mu^{\tau},\sigma^{\tau}$ {(indexed by a ``size factor'' $\tau>0$)},
%{\color{blue}(later we will take $(\mu^{\tau},\sigma^{\tau}):=(\bar{\mu}+\tau \Delta\mu,\bar{\sigma}+ \tau \Delta \sigma)$)}
there are at least two modelling options. One, which we call the \textit{strongly perturbed formulation}, is to consider a new process $S^{\tau}$ with dynamics like that of $S$ but under the new parameters, so that the {perturbed} wealth processes {have} the form:
\be\ba{rcl}
\dd X_t^{\pi,\tau}&=&  \pi_t^{\top} \mu_t^{\tau}  \dd t+\pi_t^{\top} \sigma_t^{\tau} \dd W_t \hspace{0.4cm} \mbox{for $ t\in [0,T]$},  \\[4pt]
X_0^{\pi,\tau}&=&  x. \ea\label{Prob(k)}
\ee
The perturbed problem becomes (we use the \textit{s} to denote \textit{strongly perturbed})
\begin{equation}
u^s(\mu^{\tau},\sigma^{\tau}):= \sup_{\pi\in\Pi}\EE\left[U(X^{\pi,{\tau}}_T) \right ].\label{eqstrong}
\end{equation}
%{\color{red}where $\K^{\tau}$ is defined as $\K$ but replacing $X^{\pi}$ by $X^{\pi,\tau}$. Some sensitivity results in this framework have been proven in \cite{prepLarsen}  }
{Now, let us assume that $\bar{\sigma}$ has   full {rank} almost everywhere and that $(\bar{\sigma} \bar{\sigma}^{\top})^{-1}$ is essentialy bounded}. Defining the {\it market price of risk} process   $$\textstyle\bar{\lambda} :=\bar{\sigma}^{\top}(\bar{\sigma}\bar{\sigma}^{\top})^{-1}\bar{\mu}\in (L^{\infty, \infty}_{\FF})^{d},$$ equation  \eqref{vdebil} can be written as
\be\label{vdebil2}\ba{rcl}
\dd X_t^{\pi}&=& \pi_t^{\top}\bar{\sigma}_t\left[\bar{\lambda}_t \dd t+ \dd W_t\right ] \hspace{0.2cm} \mbox{for all } \; t \in [0,T],\\[4pt]
X_0^{\pi}&=&x. \ea
\ee
{Following \cite{ZitKar}, instead of fixing the reference probability measure $\PP$ and considering perturbations directly affecting the dynamics of the $X$'s, it is reasonable to fix the latter processes (i.e.\ with the nominal parameters) and assume that the reference probability measure is perturbed. Given the perturbed parameters $(\mu^{\tau},\sigma^{\tau})$, assuming that  {$(\sigma^{\tau} (\sigma^{\tau})^{\top})^{-1}$ is essentially bounded} and setting 
$$\textstyle \lambda^{\tau}:=(\sigma^{\tau})^{\top}(\sigma^{\tau}(\sigma^{\tau})^{\top})^{-1}\mu^{\tau},$$
for the corresponding perturbed market price of risk process, its is natural to define 
 $$\textstyle\dd\PP^{\tau}:=\mathcal{E}\left ( \int(\lambda^{\tau}-\bar{\lambda})^{\top}\dd W \right)_T \dd\PP.$$
{Note that Novikov's condition implies that $\PP^{\tau}$ is a probability measure, equivalent to $\PP$}. As explained in  \cite[Section 2.2]{ZitKar}, if  $(\mu^{\tau}, \sigma^{\tau})$ converges to $(\bar{\mu},\bar{\sigma})$, then $\PP^{\tau}$ converges to $\PP$ in the total variation norm. 
%
%
%
%Regarding the perturbed parameters $(\mu^{\tau}, \sigma^{\tau})$, suppose that $\mbox{Ker}(\sigma_t^{\tau})=\mbox{Ker}(\bar{\sigma}_t)$ almost everywhere and that {$(\sigma^{\tau} (\sigma^{\tau})^{\top})^{-1}$ is essentialy bounded}. Then, $\{ \pi^{\top}\sigma \; ; \; \pi \in \Pi\}=\{  \pi^{\top}\sigma^{\tau} \; ; \; \pi \in \Pi\}$ and  {$$\textstyle \lambda^{\tau}:=(\sigma^{\tau})^{\top}(\sigma^{\tau}(\sigma^{\tau})^{\top})^{-1}\mu^{\tau},$$ is well-defined and again essentially bounded}.
 Therefore, taking this point of view, we define
%  perturbing $(\bar{\mu},\bar{\sigma})$ in the optimization problem amounts to perturbing  the reference probability measure induced by the dynamics $ \bar{\lambda}_t \dd t + \dd W_{t}$,  Girsanov's theorem motivates the definition of 
\begin{align}
u^w(\mu^{\tau},\sigma^{\tau})&:= \sup_{ \pi\in\Pi} \EE^{\PP^{\tau}}\left[U(X^{\pi}_T) \right ],\label{wvf}
\end{align}
and we call} $u^{w}$ {the} \textit{weakly perturbed formulation} of $u(\bar{\mu},\bar{\sigma})$ in \eqref{problemaformulacionfuerte}, where we insist, one modifies the initial problem by changing the probability measure. Let us remark that this function is motivated only \textit{locally} in the sense that $\bar{\mu}$ and $\bar{\sigma}$, which determine $X^{\pi}$ for a given $\pi\in\Pi$, have been fixed in order to define it. We omit this dependence from the notation of $u^w$. Of course {$u^{s}(\bar{\mu}, \bar{\sigma})=u^{w}(\bar{\mu}, \bar{\sigma})= u(\bar{\mu},\bar{\sigma})$. For the sake of clarity, we fix now the assumptions made for $(\bar{\mu}, \bar{\sigma})$ and the perturbed parameters $(\mu^{\tau}, \sigma^{\tau})$:\medskip\\
{{\bf(H1)} The matrix $\bar{\sigma}$ has full rank and $(\bar{\sigma}\bar{\sigma}^{\top})^{-1}$ is uniformly bounded in $(t,\omega)$. Moreover,  the perturbations $\sigma^{\tau}$ of $\bar{\sigma}$ satisfy  $\mbox{Ker}(\bar{\sigma})=\mbox{Ker} (\sigma^{\tau})$ (equivalently $\mbox{Im}(\bar{\sigma}^{\top}) =\mbox{Im}([\sigma^{\tau}]^{\top}) $) and $(\sigma^{\tau}(\sigma^{\tau})^{\top})^{-1}$ is uniformly bounded in $(t,\omega)$.} \smallskip

The weakly and strongly perturbed value functions in terms of the $\lambda^{\tau}$'s are defined by overloading notation: $u^w(\lambda^{\tau}):=u^w(\mu^{\tau},\sigma^{\tau})$ and $u^s(\lambda^{\tau}):=u^s(\mu^{\tau},\sigma^{\tau})$. We remark that under {\bf(H1)} the strongly perturbed value function $u^s(\mu^{\tau},\sigma^{\tau})$ coincides with the one presented in \cite{prepLarsen}. In fact, noting that {\bf(H1)} implies that 
$$ \textstyle \left\{\int_{0}^{T}\hat{\pi}_t^{\top}   [\lambda^{\tau} \dd t + \dd W_t ]\; ; \; \hat{\pi}= [\sigma^{\tau}]^{\top}\pi \; , \; \pi\in\Pi\right\} = \left\{\int_{0}^{T}{\pi}_t^{\top}   [\bar{\sigma}\lambda^{\tau} \dd t + \bar{\sigma}\dd W_t ]\; ; \; \pi\in\Pi\right\},$$
setting
 $\tilde{\lambda}^{\tau}:=(\bar{\sigma}\bar{\sigma}^{\top})^{-1}\bar{\sigma}\lambda^{\tau}$ we get
\begin{align*}
 u^s(\mu^{\tau},\sigma^{\tau})&= \textstyle \sup\left\{ \EE\bigl[U\bigl(x+\int_{0}^{T}\hat{\pi}_t^{\top}   [\lambda^{\tau} \dd t + \dd W_t ]\bigr )\bigr] \; ; \; \hat{\pi}= [\sigma^{\tau}]^{\top}\pi \; \; \mbox{for some $\pi\in \Pi$} \right\}\\
 &= \textstyle\sup\left\{ \EE\bigl[U\bigl(x+\int_{0}^{T} \pi_t^{\top}   [\bar{\sigma}\lambda^{\tau} \dd t + \bar{\sigma}\dd W_t \bigr )\bigr] \; ; \; \pi \in \Pi \right\}\\
  &=\textstyle \sup\left\{ \EE\bigl[U\bigl(x+\int_{0}^{T} \pi_t^{\top}   [\bar{\sigma}\bar{\sigma}^{\top}\tilde{\lambda}^{\tau} \dd t + \bar{\sigma}\dd W_t \bigr )\bigr] \; ; \; \pi \in \Pi \right\}\\
  &= \textstyle\sup\left\{ \EE\bigl[U\bigl(x+\int_{0}^{T} \pi_t^{\top}   [\dd \langle M\rangle_t  \tilde{\lambda}^{\tau}+\dd M_t ]\bigr )\bigr] \; ; \; \pi \in \Pi \right\}=:\tilde{u}(\tilde{\lambda}^{\tau}),
 \end{align*}
where $M_t:=\int_0^t \bar{\sigma}\dd W$. Therefore,  we can interpret $M$ as the (unperturbed) martingale driving the market in \cite{prepLarsen} and $\tilde{\lambda}^{\tau}$ as the corresponding market price of risk, which one may vary, and hence $\tilde{u}(\tilde{\lambda}^{\tau})$ is a perturbed value function of its own. If however $\mbox{Ker}(\bar{\sigma})=\mbox{Ker}(\sigma^{\tau})$ fails, both the approach of \cite{prepLarsen} as well as our approach pertaining $u^w$ are ill-suited. 
{\begin{remark} \label{kernelstability}
{\rm(i)}  From the previous discussion we see that the sensitivity analysis of $u^w$ is meaningful under the condition {\bf(H1)} on the Kernels, in which case also the study of $\tilde{u}$ above makes sense.  {This invariance of the null space of the volatility term under the considered perturbations is our main assumption and allows us to provide explicit sensitivity results in terms of perturbations of the volatility term $\bar{\sigma}$. {Let us point out that in the complete case (i.e.\ $\bar{\sigma}$ is invertible)  a similar argumentation can be found in   \cite[Section 2.2]{ZitKar}.} The case of general perturbations of $\bar{\sigma}$ is beyond the scope of the present work; see \cite{Weston} for an insight into the difficulties to be expected.} \\ 
%{\color{red}The  invariance of $\mbox{{\rm Ker}} (\sigma^{\tau})$ in {\bf(H1)} ensures that $\{ \pi^{\top}\sigma \; ; \; \pi \in \Pi\}=\{  \pi^{\top}\sigma^{\tau} \; ; \; \pi \in \Pi\}$, which is a fundamental property allowing to prove that $u^{s}(\mu^{\tau},\sigma^{\tau})=u^{w}(\mu^{\tau},\sigma^{\tau})$, at least when the nominal and perturbed parameters are deterministic (see Proposition \ref{igualdadfuerteydebil}).}\\ 
{\rm(ii)}The assumptions for  $\sigma^{\tau}$ in {\bf(H1)} are satisfied for  $\sigma^{\tau}= \bar{\sigma}+ A^{\tau} (\bar{\sigma} \bar{\sigma}^{\top})^{-1} \bar{\sigma}$ where $A^{\tau} \in (L^{\infty,\infty}_{\FF})^{d\times d}$   has small enough norm. This holds in particular for $\sigma^{\tau}=\bar{\sigma}+\tau A (\bar{\sigma} \bar{\sigma}^{\top})^{-1} \bar{\sigma}$ with $A\in (L^{\infty,\infty}_{\FF})^{d\times d}$ arbitrary and $\tau$ a small enough real number. 
\end{remark}
}

As we will see in Section \ref{subsec counter}, the values $u^{s}$ and $u^{w}$, as well as their sensitivities, generally differ. 
%We will provide two examples of this phenomenon in Section \ref{onvaluefunctions}. 
%The reason for this discrepancy seems to be that we have kept the filtration fixed as we perturbed the parameters, a common practice when studying sensitivity and robustness in stochastic optimization. {\color{blue}JB: heuristicamente es razonable q el no perturbar la filtracion tenga esta consecuencia ... por otro lado, en el Ejemplo 1 la filtracion efectivamente no cambia con la perturbacion y aun asi difieren las formulaciones fuerte y debil. Decimos esto? omitimos la frase anterior? dejamos las cosas como estan? }
On the other hand, the next {result} shows that if the parameters $\bar{\mu}$, $\bar{\sigma}$ and their perturbations $\mu^{\tau}$, $\sigma^{\tau}$ are deterministic, then $u^{s}$ and $u^{w}$ (and so their sensitivities)  do coincide.
% 
%{\begin{proposition}\label{igualdadfuerteydebil}
%Suppose that $(\Omega,\PP)$ is the canonical space of continuous  paths in $\RR^{d}$ equipped with the Wiener measure. {Assume that $\bar{\mu}$, $\bar{\sigma}$ and the perturbed parameters  $\mu^{\tau}$, $\sigma^{\tau}$ are all deterministic and  {\bf(H1)} holds true. 
%%that $\mbox{{\rm Ker}}(\bar{\sigma})=\mbox{{\rm Ker}}(\sigma^{\tau})$ and that both  $(\bar{\sigma}\bar{\sigma}^{\top})^{-1}$ and $({\sigma}^{\tau}[{\sigma}^{\tau}]^{\top})^{-1}$ are essentially bounded. 
%Then, $u^{s}(\mu^{\tau},\sigma^{\tau})=u^{w}(\mu^{\tau},\sigma^{\tau})$}
%%
%%\begin{enumerate}
%%\item $\mbox{Ker}(\bar{\sigma})=\mbox{Ker}(\sigma^{\tau})$ ,
%%\item $n=d$.
%%\end{enumerate}
%% and that $\bar{\mu}, \bar{\sigma},\mu^{\tau}$ are deterministic and that $\sigma^{\tau}=\bar{\sigma}$. Then   
%%as long as coefficients and portfolio weighs are taken previsible.
%\end{proposition}}
%%{\color{blue} Lo anterior es incorrecto, segun veo ahora. Creo que con la condicion de los Ker, nuestra formulacion fuerte en terminos de mu y sigma siempre coincide con la de Larsen et al. Yo lo enunciaria asi:\\

{\begin{proposition}\label{igualdadfuerteydebil}
{Assume that $\mu^{\tau},\sigma^{\tau},\bar{\mu},\bar{\sigma}$ are deterministic, and that {\bf(H1)} holds.
%, and that the following stochastic differential equation has pathwise-uniqueness:
%%
%\begin{align}
%\textstyle \label{sol trayectorial}
%\dd R_t &= \textstyle [\lambda^{\tau}(R,t) - \bar{\lambda}(R,t)]\dd t +dB_t,
%\end{align}
%% 
%which is true e.g.\ if $\bar{\mu}$, $\bar{\sigma}$ are deterministic too. 
%that $\mbox{{\rm Ker}}(\bar{\sigma})=\mbox{{\rm Ker}}(\sigma^{\tau})$ and that both  $(\bar{\sigma}\bar{\sigma}^{\top})^{-1}$ and $({\sigma}^{\tau}[{\sigma}^{\tau}]^{\top})^{-1}$ are essentially bounded. 
Then the weak and strong value functions coincide; $u^{s}(\mu^{\tau},\sigma^{\tau})=u^{w}(\mu^{\tau},\sigma^{\tau})$.}
%
%\begin{enumerate}
%\item $\mbox{Ker}(\bar{\sigma})=\mbox{Ker}(\sigma^{\tau})$ ,
%\item $n=d$.
%\end{enumerate}
% and that $\bar{\mu}, \bar{\sigma},\mu^{\tau}$ are deterministic and that $\sigma^{\tau}=\bar{\sigma}$. Then   
%as long as coefficients and portfolio weighs are taken previsible.
\end{proposition}

\begin{proof}

Define $B_t:=W_t-\int_0^t[\lambda^{\tau}_s-\bar{\lambda}_s]\dd s$, so by Girsanov Theorem $B$ is a $\PP^{\tau}$-Brownian motion. Notice that $\F^B=\F^W$. Taking $\pi$ feasible for the perturbed problem we have
\be\label{calculo_equivalencia}
\ba{rl}\textstyle
\EE^{\PP}\bigl[U\bigl( x+ \int_0^T\pi_t(W)^{\top}\sigma_t^{\tau}[\lambda_t^{\tau}\dd t +\dd W_t  ]  \bigr)  \bigr ] &= \textstyle \EE^{\PP^{\tau}}\bigl[U\bigl( x+ \int_0^T\pi_t(B)^{\top}\sigma_t^{\tau}[{\lambda}_t^{\tau}\dd t +\dd B_t  ]  \bigr)  \bigr ]  \\[4pt]
&= \textstyle \EE^{\PP^{\tau}}\bigl[U\bigl( x+ \int_0^T\pi_t(B)^{\top}\sigma_t^{\tau}[\bar{\lambda}_t\dd t +\dd W_t  ]  \bigr)  \bigr ] \\[4pt]
&= \textstyle \EE^{\PP^{\tau}}\bigl[U\bigl( x+ \int_0^T\tilde{\pi}_t(W)^{\top}{\sigma}_t^{\tau}[\bar{\lambda}_t\dd t +\dd W_t  ]  \bigr)  \bigr ]\\[4pt]
&= \textstyle \EE^{\PP^{\tau}}\bigl[U\bigl( x+ \int_0^T\hat{\pi}_t(W)^{\top}\bar{\sigma}_t[\bar{\lambda}_t\dd t +\dd W_t  ]  \bigr)  \bigr ]\\[4pt]
&\textstyle \leq u^w({\mu}^\tau,{\sigma}^\tau),
\ea 
\ee
where we first used that $B$ is $\PP^{\tau}$-BM, then the definition of $B$, then we built $\tilde{\pi}$ by equality of filtrations, and finally the assumption on the image of the matrices $\bar{\sigma}^{\top}$ and $[\sigma^{\tau}]^{\top}$. %Since $\PP( x+ \int_0^T\pi_t(W)^{\top}\sigma_t^{\tau}[\lambda_t^{\tau}\dd t +\dd W_t])\geq 0)=1$, the same arguments lead to $\PP^\tau(x+ \int_0^T\hat{\pi}_t(W)^{\top}\bar{\sigma}_t[\bar{\lambda}_t\dd t +\dd W_t  ]\geq 0 )=1$
 Having begun with a feasible element for the unperturbed problem and reasoning as above, yields the opposite inequality.
\end{proof}

}
{
\begin{remark}\label{difficulty_stochastic_parameters}
Note that if  $\mu^{\tau},\sigma^{\tau},\bar{\mu},\bar{\sigma}$ are random, then the previous proof does not work. Indeed, following the lines of the proof, we would have that $B_t:=W_t-\int_0^t[\lambda^{\tau}_s(W)-\bar{\lambda}_s(W)]\dd s$ is a $\PP^{\tau}$-Brownian motion and so, following \eqref{calculo_equivalencia}, we would get
$$\textstyle
\EE^{\PP}\bigl[U\bigl( x+ \int_0^T\pi_t(W)^{\top}\sigma_t^{\tau}(W)[\lambda_t^{\tau}(W)\dd t +\dd W_t  ]  \bigr)  \bigr ] = \textstyle \EE^{\PP^{\tau}}\bigl[U\bigl( x+ \int_0^T\pi_t(B)^{\top}\sigma_t^{\tau}(B)[{\lambda}_t^{\tau}(B)\dd t +\dd B_t  ]  \bigr)  \bigr ], 
$$
whose right hand side generally differs from 
$$\textstyle \EE^{\PP^{\tau}}\bigl[U\bigl( x+ \int_0^T\pi_t(B)^{\top}\sigma_t^{\tau}(B)[\bar{\lambda}_t(W)\dd t +\dd W_t  ]  \bigr)  \bigr ].$$
%{\color{red}On top of this is the fact that only $\F^B\subset\F^W$ holds in general.}
\end{remark}
}

\vspace{0.5cm}

Let us go back, for once and for all, to weakly perturbed parameters. As commented in the introduction, the continuity of $u^w$ (in a broader context) as a function of $\lambda$ was analysed in \cite{ZitKar} . We move towards the first-order analysis now. Consider the set
\begin{equation}
\mathcal{M}^e(S)=\left\{\PP^* \sim \PP:S \mbox{ is a }\PP^*\mbox{-local martingale}  \right\} .\label{me}
\end{equation}
{By \cite[Proposition 7.2.1]{Pham} we have that $\mathcal{M}^e(S)$ is given by the set  of random variables $Y^{\nu}_T$, where for $ \nu^{i}\in L^2_{loc}(W^{i})$ ($i=1,\hdots,m$) and  $\nu \in \mbox{Ker}(\bar{\sigma})$ almost everywhere, and where the process $Y^{\nu}_t$ is the exponential martingale  ${Y^{\nu}_t}:=\E\left (-\int [\bar{\lambda}+\nu]^{\top}\dd W \right )_t$}.  Given $Z\in L^{0}$, let us define
\begin{equation} J(Z):= \sup_{\mathbb{M} \in \mathcal{M}^e(S)}\EE^{\mathbb{M}}\left [U^{-1}(\vert Z\vert) \right ]. \label{Jinicial}
\end{equation}
Since $\sup_{ \pi\in\Pi} \EE^{\PP^{\tau}}\left[U(X^{\pi}_T) \right ]=\sup_{L \in C(x)} \EE^{\PP^{\tau}}\left[U( L) \right]$,
where 
$$C(x)=\left\{ L \in  L^{0}_{+} \; ; \; \exists \;  \pi\in\Pi, \; \; L\leq  X_{T}^{\pi}  \; \; \mbox{a.s.}\right\},$$
letting   {$Z= U(X_T)$ and using the usual budget-constraint (see e.g.\ \cite[Corollary 7.2.1]{Pham})} we can further rewrite problem \eqref{wvf} as:
\begin{equation}\textstyle
{u^w(\mu^{\tau},\sigma^{\tau})=u^{w}(\lambda^{\tau})}=\sup \left\{\EE\bigl[\mathcal{E}\bigl ( \int(\lambda^{\tau}-\bar{\lambda)}^{\top}\dd W \bigr)_T Z\bigr] \; ; \; J(Z)\leq x, \; \;  Z\in L^0_+ \right\}.\label{lamano}
\end{equation}
%
%\begin{equation}
%u^w(\mu^{\tau},\sigma^{\tau})= \sup\limits_{\substack{Z\in (L_J)^+\\ J(Z)\leq x}}\EE^{\PP}\left [\exp\left \{ \int_0^T(\lambda^{\tau}-\lambda)(t)^{\top}\dd W(t) -\frac{1}{2}\int_0^T|(\lambda^{\tau}-\lambda)(t)|^2\dd t \right\} Z\right] ,\label{lamano}
%\end{equation}
\normalsize

 Thanks to our rewriting of $u^w$ in \eqref{lamano}, we will be able to deal with the analysis of the differentiability of this function with respect to all the parameters. { More precisely, \eqref{lamano} opens the way to interpreting the sensitivity analysis of $u^w$ as the study of a convex-theoretic support function, as we had hinted at in the introduction.}
%Independently of us, the sensitivity analysis of a value function related to $u^s$ is currently being carried out in \cite{prepLarsen}, as we have already hinted.
%In this article we perform a \textit{sensitivity analysis in the weak perturbations framework}. From e.g.\ \eqref{lamano} we see that this setting has the advantage that the perturbed parameters disappear from the controlled equation (the constraint) and only enter in the problem through the stochastic exponential. 
Under appropriate assumptions, we ultimately prove in Theorem \ref{teoremarkmenossimpli} the following sensitivity results with respect to $(\mu,\sigma)$. We refer the reader to Definition \ref{INADA} for the meaning of $U$ being a utility function satisfying INADA conditions, and to the appendix for the definition of Hadamard differentiability:
\begin{theorem}\label{teoremarkmenossimpli} 
Suppose $U$ is an utility function satisfying INADA conditions and  such that $U(0+)=0$ as well as the bound for some $p\in(1,\infty)$:
$$U(x)\,\leq\, C x^{1/p},\,\, \mbox{ for all }x\geq 0.$$
Consider some perturbations $(\Delta \mu, \Delta \sigma) \in  (L^{\infty, \infty}_{\FF})^{d}\times (L^{\infty, \infty}_{\FF})^{d\times n}$ and suppose that {\bf(H1)} is satisfied for  $(\mu^{\tau}, \sigma^{\tau}):=(\bar{\mu}+ \tau \Delta \mu, \bar{\sigma}+\tau \Delta\sigma)$  and small enough $\tau$. Then, the directional derivative $Du^{w}(\bar{\mu},\bar{\sigma})(\Delta \mu, \Delta \sigma)$ exists and is given by
%\footnotesize
%$$\mathcal{P}:=\left\{({\mu},{\sigma})\in (L^{\infty, \infty}_{\FF})^{d}\times (L^{\infty, \infty}_{\FF})^{d\times n} : \sigma\sigma^{\top} \mbox{ is a.e. invertible and } \esssup_{t,\omega} | [\sigma\sigma^{\top}]^{-1}|<\infty\right\}.$$ 
%\normalsize
%Fix $(\bar{\mu},\bar{\sigma})\in \mathcal{P}$ and define the function $u^w$ around it as in \eqref{lamano}. Then $u^{w}$ is continuous, G\^ateaux and directionally Hadamard differentiable at $(\bar{\mu}, \bar{\sigma})\in \mathcal{P}$. The directional derivatives are given by:
%
%
\begin{align*}\textstyle
%D_ru(\bar{r},\mu,\sigma)\Delta r &=& \EE \left [ U(\bar{X}(T))L \left\{ -\int_0^T [\sigma^{\top}\{\sigma\sigma^{\top}\}^{-1}\Delta r{\bf 1}]^{\top}\dd W + \int_0^T \langle \lambda-\bar{\lambda},\sigma^{\top}\{\sigma\sigma^{\top}\}^{-1}\Delta r{\bf 1} \rangle\dd t \right\} \right ] \\
D_{\mu}u^{{w}}(\bar{\mu},\bar{\sigma})\Delta \mu &=\textstyle \EE \left [ U(\bar{X}(T))  \int_0^T [\bar{\sigma}^{\top}[\bar{\sigma}\bar{\sigma}^{\top}]^{-1}\Delta \mu]^{\top}\dd W \right ], \\[4pt]
D_{\sigma}u^{{w}}(\bar{\mu},\bar{\sigma})\Delta \sigma &= \textstyle \EE \left [ U(\bar{X}(T))   \int_0^T \left\{\Delta\sigma^{\top}[\bar{\sigma}\bar{\sigma}^{\top}]^{-1}\bar{\mu}-  \bar{\sigma}^{\top} [\bar{\sigma}\bar{\sigma}^{\top}]^{-1}  [\bar{\sigma}\Delta\sigma^{\top} +\Delta\sigma\bar{\sigma}^{\top} ]  [\bar{\sigma}\bar{\sigma}^{\top}]^{-1}\bar{\mu}\right\}^{\top}\dd W  \right],
\end{align*}
where $\bar{X}(T)$ is the unique optimal terminal wealth attaining  $u(\bar{\mu},\bar{\sigma})$. Moreover, the application $(\mu, A) \in  (L^{\infty, \infty}_{\FF})^{d}\times (L^{\infty, \infty}_{\FF})^{d\times d}\mapsto  u^{w}(\mu, A (\bar{\sigma}\bar{\sigma}^{\top})^{-1}\bar{\sigma}) \in \RR$ is Hadamard differentiable at $(\bar{\mu},\bar{\sigma}\bar{\sigma}^{\top})$.
\end{theorem}
An example of $U$ satisfying the assumptions in Theorem \ref{teoremarkmenossimpli} is $U(x)=x^{1/p}$ with $p\in(1,\infty)$, the so-called positive power case. A further example  is given e.g.\ by the inverse function of $y\in [0,\infty)\mapsto R(y):=e^y-y-1$. Indeed, $R^{-1}$ is non-negative, strictly concave and increasing, with $R^{-1}(0)=0$. It is also differentiable in $(0,\infty)$ and from $[R^{-1}]'(x)=1/(R'\circ R^{-1}(x))$ we find that $[R^{-1}]'(0)=+\infty$ and $[R^{-1}]'(+\infty)=0$. Finally, we easily see that $R^{-1}(x)\leq \sqrt{2} x^{1/2}$, or equivalently $y^2\leq 2[e^y-y-1]$, by Taylor expansion. Our result does not cover the case of negative powers.\\
We finally remark that if the market defined by $(\bar{\mu},\bar{\sigma})$ is complete, then $n=d$ and $\bar{\sigma}$ is invertible (see e.g.\ \cite[Theorem 6.6, Chapter 1]{KaraShreveFinance}). In this case, $u^{w}$ is Hadamard differentiable at $(\bar{\mu}, \bar{\sigma})$ and 
\be\label{diffcasocompleto}\textstyle Du^{w}(\bar{\mu},\bar{\sigma}) (\Delta \mu, \Delta \sigma) = \EE \left [ U(\bar{X}(T))  \int_0^T [\bar{\sigma}^{-1}  \Delta \mu- \bar{\sigma}^{-1}\Delta \sigma \bar{\sigma}^{-1}\bar{\mu}]^{\top}\dd W_t \right ].\ee

We proceed now to the counterexamples promised before Proposition \ref{igualdadfuerteydebil} and in the introduction.

\subsection{Counterexamples}\label{subsec counter}

Let us illustrate how, even in the  one-dimensional case,  $u^s$ and $u^w$ (as well as their directional derivatives) generally defer. For this to be the case, it is important that the reference market price of risk $\bar{\lambda}$ be random.

\begin{proof}[Example 1]
Let us take $U(x)=\log(x)$ if $x>0$ and $U(x)=-\infty$ if $x\leq 0$. Although this utility function does not fulfil our assumption, we use it to illustrate the phenomenon we are discussing. It is well known (see e.g.\ \cite[Chapter 7.3.5]{Pham}) that for a market model $\dd S_t=\lambda\dd \langle M\rangle_t  +\dd M_t$ for $M$ a martingale and $\lambda$ say essentially bounded, 
%the optimizer is $$x\exp\left\{ -\int_0^T \lambda\dd M - \frac{1}{2}\int_0^T \lambda^{\top}\dd\langle M\rangle\lambda \right\}$$
the optimal utility is 
$$\textstyle \log(x)+\frac{1}{2}\EE\bigl[ \int_0^T \lambda_t^{\top}\dd\langle M\rangle_t\lambda_t \bigr].$$
We thus conclude in our Brownian setting and for $\lambda^{\tau}=\bar{\lambda}+\tau\Delta$ that:
\begin{align*}\textstyle
u^s(\lambda^{\tau})& = \textstyle\log(x)+\frac{1}{2}\EE\bigl[\int_0^T\vert\bar{\lambda}+\tau\Delta\vert^2\dd t\bigr], \\
&=\textstyle \log(x)+\frac{1}{2}\EE\bigl[\int_0^T\vert\bar{\lambda}\vert^2\dd t\bigr]+\tau \EE\bigl[\int_0^T \bar{\lambda}^{\top}\Delta \dd t\bigr]+\frac{\tau^2}{2}\EE\bigl[\int_0^T\vert\Delta\vert^2\dd t\bigr].
\end{align*}
On the other hand, denoting $\dd\PP^{\tau}= \E\left(\tau\int\Delta\dd W \right)_T\dd\PP$ so $ W^{\tau}=W-\tau\int \Delta \dd t$ is a {$\PP^{\tau}$}-Brownian motion by Girsanov's theorem, and taking $\Delta$ deterministic so that $\F^{W^{\tau}}=\F$,  we get
\begin{align*}\textstyle
u^w(\lambda^{\tau})& =\textstyle \log(x)+\frac{1}{2}\EE^{\PP^{\tau}}\bigl[\int_0^T\{\bar{\lambda}+\tau\Delta\}^2\dd t\bigr],\\
&=\textstyle\log(x)+\frac{1}{2}\EE^{\PP^{\tau}}\bigl[\int_0^T\vert\bar{\lambda}\vert^2\dd t\bigr]+\tau \EE^{\PP^{\tau}}\bigl[\int_0^T \bar{\lambda}^{\top}\Delta \dd t\bigr]+\frac{\tau^2}{2}\EE^{\PP^{\tau}}\bigl[\int_0^T\vert\Delta\vert^2\dd t\bigr].
\end{align*}
 This already shows that the two value functions may easily differ, unless e.g.\ $\bar{\lambda}$ were further deterministic. Moreover, one can easily compute the first order sensitivities:
\begin{align*}\textstyle
\left .\frac{\dd u^s(\lambda^{\tau})}{\dd\tau}\right\vert_{\tau=0} &= \textstyle\EE\bigl[\int_0^T \bar{\lambda}^{\top}\Delta \dd t\bigr],\\ \textstyle
\left .\frac{\dd u^w(\lambda^{\tau})}{\dd\tau}\right\vert_{\tau=0} &= \textstyle\EE\bigl[\int_0^T \bar{\lambda}^{\top}\Delta \dd t\bigr] + \frac{1}{2}
\EE\bigl[\int_0^T \vert\bar{\lambda}\vert^2 \dd t\int_0^T\Delta\dd W\bigr],\\
&=\textstyle \EE\bigl[\int_0^T \bar{\lambda}^{\top}\Delta \dd t\bigr] + \frac{1}{2}
\EE\bigl[\int_0^T\bigl\{\int_0^t\Delta_s\dd W_s \bigr\} \vert\bar{\lambda}_t\vert^2 \dd t\bigr].
\end{align*}
We conclude that the sensitivities generally differ, unless again if e.g.\ $\bar{\lambda}$ was deterministic. To exemplify this point, the reader may take any bounded deterministic function {$\Delta$} and define $\bar{\lambda}(t,\omega)$ to be e.g.\ of euclidean norm $1$ if $\int_0^t\Delta_s\dd W_s$ is positive and $0$ otherwise.

This example also shows that the weak value function can behave in a counter-intuitive way in the presence of random parameters. For instance, taking $\bar{\lambda}_t:= {\bf 1}_{W_t<0}$ and $\Delta\equiv 1$ it is elementary to see that 
$$\textstyle \left .\frac{\dd u^s(\bar{\lambda} +{\tau})}{\dd\tau} \right\vert_{\tau=0} = \frac{T}{2} \,\,\,\,\,\,\,\, \mbox{ and } \,\,\,\,\,\,\,\,  {\left.\frac{\dd u^w(\bar{\lambda} +{\tau})}{\dd\tau} \right\vert_{\tau=0}} = \frac{T}{2} - \frac{T^{3/2}}{3\sqrt{2\pi}},$$
so as intuition suggest utility increases in the strong formulation whereas (for $T$ large enough) it decreases in the weak one.

%
%as $W-\tau\int\Delta\dd t$ is a brownian motion under
%{\color{red} aqui me parece que se tiene que usar que $\Delta$ es determinista para tener la misma filtracion browniana y usar la funcion valor ya conocida (misma historia de antes) marco los delta en rojo abajo en referencia a su hipotesis }\\
%{\color{blue} Creo que en general tendrias razon, no me habia percatado, pero en el caso logaritmico pasa que la solucion es siempre el reciproco de $\E(-\int\lambda\dd M)$, ver e.g.\cite[Section 2.6]{ZitLar}. Igual, pa no arriesgarnos y hacernos mas facil la vida, tomemos el $\Delta$ determinista no mas... ya lo cambi\'e en ambos ejemplos. }
\end{proof}\smallskip
\begin{proof}[Example 2]
We now present an example that does fulfil our assumptions on the utility function. Let us take $U(x)=2\sqrt{x}$ if $x\geq 0$ and $-\infty$ otherwise. We take $x=1$ for simplicity. By e.g.\ \cite[Chapter 7.3.5]{Pham} we know, in the one-asset case, that the optimal utility for a market model $\dd S=\lambda\dd \langle M\rangle+\dd M$ will be $$ \textstyle2\sqrt{\EE\bigl [ \exp\bigl\{ \int_0^T \lambda\dd M +\frac{1}{2}\int_0^T \lambda^2\dd\langle M\rangle\bigr\}\bigr]}.$$
Thus, in a one-dimensional Brownian setting and for $\lambda^{\tau}=\bar{\lambda}+\tau\Delta$ it holds:
\begin{align*} \textstyle
u^s(\lambda^{\tau})& = \textstyle 2\sqrt {\EE\bigl[ \exp\bigl\{ \int_0^T[ \bar{\lambda}+\tau\Delta]\dd W +\frac{1}{2}\int_0^T [ \bar{\lambda}+\tau\Delta]^2\dd t\bigr\}\bigr ] },
\end{align*}
and by Girsanov's theorem {and assuming $\Delta$ deterministic}:
\begin{align*} \textstyle
u^w(\lambda^{\tau})& =2 \textstyle\sqrt {\EE^{\PP^{\tau}}\bigl [ \exp\bigl\{ \int_0^T[ \bar{\lambda}+\tau\Delta](\dd W-\tau\Delta\dd t) +\frac{1}{2}\int_0^T [ \bar{\lambda}+\tau\Delta]^2\dd t\bigr\}\bigr] },
\end{align*}
where $\dd\PP^{\tau}= \E\left(\tau\int\Delta\dd W \right)_T\dd\PP$. We thus obtain the following first order sensitivities:
\begin{align*} \textstyle
\left .\frac{\dd u^s(\lambda^{\tau})}{\dd\tau}\right\vert_{\tau=0} &= \textstyle\frac{\EE\bigl[e^{\int_0^T\bar{\lambda}\dd W + \frac{1}{2}\int_0^T\bar{\lambda}^2\dd t}[ \int_0^T\Delta\dd W+\int_0^T\Delta\bar{\lambda}\dd t ] \bigr]}{\sqrt{\EE\left[e^{\int_0^T\bar{\lambda}\dd W + \frac{1}{2}\int_0^T\bar{\lambda}^2\dd t}\right ]}},\\ \textstyle
\left .\frac{\dd u^w(\lambda^{\tau})}{\dd\tau}\right\vert_{\tau=0} &= \textstyle  2\frac{\EE\bigl[e^{\int_0^T\bar{\lambda}\dd W + \frac{1}{2}\int_0^T\bar{\lambda}^2\dd t}\int_0^T\Delta\dd W \bigr]}{\sqrt{\EE\left[e^{\int_0^T\bar{\lambda}\dd W + \frac{1}{2}\int_0^T\bar{\lambda}^2\dd t}\right ]}}.
\end{align*}
From this, we see that 
$$\textstyle \left .\frac{\dd u^s(\lambda^{\tau})}{\dd\tau}\right\vert_{\tau=0}= \left .\frac{\dd u^w(\lambda^{\tau})}{\dd\tau}\right\vert_{\tau=0}\iff \EE\bigl[e^{\int_0^T\bar{\lambda}\dd W + \frac{1}{2}\int_0^T\bar{\lambda}^2\dd t}\bigl (\int_0^T\Delta\dd W - \int_0^T\Delta\bar{\lambda}\dd t\bigr ) \bigr]=0. $$\normalsize
This shows that the sensitivities generally differ, unless if further e.g.\ $\bar{\lambda}$ is deterministic. To exemplify, with Girsanov theorem and the product formula, the expectation in the r.h.s\ above becomes
$$ \textstyle\tilde{\EE}\left [ \int_0^T\left ( \int_0^t\Delta_s\dd {W}_s -\int_0^t \Delta_s\bar{\lambda}_s\dd s \right )e^{\int_0^t\bar{\lambda}^2_s\dd s}\bar{\lambda}^2_t\dd t\right ],$$ 
where $\tilde{\EE}$ denotes expectation under $\dd\tilde{\PP}:=\E\left(\int \bar{\lambda}\dd W\right)_T\dd \PP$. 
%and $\tilde{W}= W-\int\bar{\lambda}\dd t$.
The reader may take any negative, bounded function {$\Delta$} and define $\bar{\lambda}(t,\omega)$ to be e.g.\ equal to $1$ if $\int_0^t\Delta_s\dd W_s$ is positive and $0$ otherwise. Then $\bigl ( \int_0^t\Delta_s\dd {W}_s -\int_0^t \Delta_s\bar{\lambda}_s\dd s \bigr )\bar{\lambda}^2_t$ is non-negative a.e.\ and can be seen to be strictly positive in a non-evanescent set. Thus the sensitivities differ in this case, and a fortriori also the value functions themselves.
\end{proof}

\section{The utility maximization problem as a support function of a weakly compact set}
\label{utilmax}
%%%%%%
{In this section we survey some of the results in \cite{BFrobustez1}, where the setting, similar to that of \cite{KrSch},  is   more general  than ours as described  in the previous section.}

Let there be $d$ stocks and a bond, normalized to one for simplicity. Let $S=\left(S^i\right)_{1 \leq i \leq d}$ be the price process of these stocks, and $T<\infty$ a finite deterministic investment horizon. The process $S$ is assumed to be a {continuous} semimartingale in a filtered probability space $(\Omega,\mathbb{F},(\mathcal{F}_{t})_{t\leq T},\PP)$, where  $\PP$ will always stand for the \textit{reference measure}. The expectation with respect to $\PP$ will be denoted by  $\EE$ as before. %The set of all probability measures on $(\Omega,\mathbb{F})$ absolutely continuous w.r.t. $\PP$ will be denoted by ${\cal P}$, and  the expectation with respect to $\QQ\in {\cal P} \backslash\{ \PP\}$ will be denoted by $\EE^{\QQ}$.

A (self-financing) portfolio $\pi$ is defined as a couple $(X_0,H)$, where $X_0\geq 0$ denotes the (constant) initial value associated to it and $H=(H^i)_{i=1}^d$ is a  predictable and $S$-integrable process  which represents  the number of shares of each type under possession. The {\it wealth} $X=(X_t)_{t\leq T}$ associated to a portfolio $\pi$ is defined as
\begin{equation} \textstyle
X_t=X_0+\int_0^t H_u \dd S_u \hspace{0.3cm} \mbox{for all $t\in [0,T]$,} \label{riqueza}
\end{equation}
and the  set of attainable wealths from $x$ is defined as
\begin{equation}
\mathcal{X}(x)=\left\{ X\geq 0 \; \; \mbox{a.s. in $\Omega \times [0,T]$} \; ; \; X \mbox{ as in } \eqref{riqueza} \; \mbox{ with } \; X_0 \leq x \right\}. \label{admisibles}
\end{equation}
%
%The set of equivalent local martingale measures (or risk neutral measures) associated to $S$ is 
%%
%%\begin{equation}
%%\mathcal{M}^e(S)=\left\{\PP^* \sim \PP:\mbox{ every } X \in \mathcal{X}(1) \mbox{ is a }\PP^*\mbox{-local martingale}  \right\} \label{me}
%%\end{equation}
%%%
%%which reduces to ({\color{red} referencia?})
%%
%\begin{equation}
%\mathcal{M}^e(S)=\left\{\PP^* \sim \PP:S \mbox{ is a }\PP^*\mbox{-local martingale}  \right\} .\label{me}
%\end{equation}
%%
%%if $S$ is locally bounded. 
We assume in the sequel that the market is \textit{arbitrage-free}, in the sense of NFLVR (see e.g.\ \cite{DelSch}), which implies that $\mathcal{M}^e(S)$ (defined as in \eqref{me}) is not empty. 
As usual the market model is coined \textit{complete} if $\mathcal{M}^e(S)$  is reduced to a singleton, i.e. $\mathcal{M}^e(S) =\{\PP^*\}$, {and incomplete otherwise}. 
{The following set, introduced in \cite{KrSch},  plays a central role in portfolio optimization in incomplete markets} 
\begin{equation*}
\mathcal{Y}_{\PP}(y):= \left\{Y \geq 0| Y_0=y \mbox{ , } XY \mbox{ is } \PP-\mbox{supermartingale } \forall X \in \mathcal{X}(1)  \right\}. \label{YQ}
\end{equation*} 
{The set $\mathcal{Y}_{\PP}(y)$ generalizes the set of density processes (with respect to  $\PP$)  of  risk neutral measures equivalent to it.

Now, we consider the following notion of  utility function.}
   
\begin{definition} \label{INADA}
{A function $U:\RR \rightarrow  \RR \cup \{-\infty\}$ is called a \textit{utility function} if $U(x)=-\infty$ if $x \in (-\infty,0)$ and on $[0,\infty)$ we have that $U$  is strictly increasing, strictly concave and continuously differentiable. We say that $U$  satisfies the  \textit{INADA} conditions}  {\rm (\cite{inada64})} if 
$$U'(0+){:= \lim_{x\downarrow 0} U'(x)}=\infty \hspace{0.3cm} \mbox{ {\rm and} } \hspace{0.3cm} U'(+\infty)=0\, . $$ 
%{\color{red}The  \textit{asymptotic elasticity} {\rm(}see {\rm \cite{KrSch}}{\rm)} is defined as $AE(U):= \limsup_{x \rightarrow \infty} \frac{xU'(x)}{U(x)}.$}
\end{definition}
%Such a function $U$ is extended  as $-\infty$ on $(-\infty,0]$.  \\

%Suppose  now that an agent aims to optimize {her/his} wealth by investing in a market. {Starting from an initial wealth not larger that $x> 0$, the agent aims to solve the optimization problem:} 
%({\color{red} pq directamente no toma $x$ como condicion inicial? es solo una razon matematica?}{\color{blue}R: s.p.g. se puede asumir igual a x ... no hacer esto tiene la ventaja de que el conjunto de weones factibles es solido})
%
%\begin{equation}\label{RobPortOpt}
%u(x):= \sup_{X \in \mathcal{X}(x)}  \EE\left(U%\left(X_T\right)\right),
%\end{equation}
%{The function $u$ is concave,  so that $u(x_0)<+\infty$ at some $x_0>0$ implies $u(\cdot)<+\infty$.}

%
%(a suitable  meaning can often be given to  the expectation in the  case $U$ is unbounded) which represents the situation in which she tries to maximize the expected utility. \\

As in e.g.\ \cite{KrSch}, we will make use of the {Fenchel conjugate  of $-U(-\cdot)$}, namely:
\begin{equation*}
V(y):=\sup_{x>0}[U(x)-xy] \mbox{ , } \hspace{0.2cm} \forall \; y>0. \label{V} \\
\end{equation*} 
%The following function,  commonly used in the literature to tackle problem   \eqref{RobPortOpt}, will also be relevant  here:
%
%Now, let us define 
%\be\label{v}
%u(x)&=&\sup_{X \in \mathcal{X}(x)} \EE\left(U\left(X_T\right)\right) ,\label{u}  \\
%v(y)= \inf_{ Y \in \mathcal{Y}_{\PP}(y)}\EE\left(V(Y_T) \right)	. 
%\ee
%
%Of course,  $u(x)$ is the investor's utility, when starting from an initial wealth not larger that $x> 0$.
%It is  proven in \cite[Theorem 3.1]{KrSch}  that,  whenever     $u$ is finite, the functions $u$ and $v$ are conjugate {in the following sense}:
%
%\begin{equation}\label{uQvQ}
%u(x)= \inf_{y>0} \left(v(y)+xy \right) \hspace{0.3cm} \mbox{   and    }  \hspace{0.3cm} v(y)= \sup_{x>0} \left(u(x)-xy \right).
%\end{equation}
%
%Let us now briefly summarize the main  available general results as in \cite{KrSch}. {\color{red} HACERLO: Me parece que lo anterior es demasiado detalle, si solo nos vamos a quedar con el caso browniano: a que resultados te refieres?}
%We will  denote in the sequel  by $L^0=L^0(\Omega,\PP)$  the space of measurable functions equipped with the  topology of convergence in probability, and by $L^0_+\subset L^0$ the cone of non-negative functions therein.  

In the remainder of this section, we will restrict our attention to the following setting:
\medskip\\
%\begin{assumption}
%$ $
%\label{Usupos}
{{\bf (A1)}} $U$ is an utility function satisfying INADA conditions and  such that $U(0+)=0$. 
%\end{assumption}

\begin{remark} The above assumption implies that $V\geq 0$ {and the existence of and inverse $U^{-1}: (0,\infty) \to (0,\infty)$}. {Of course,  by a translation argument we can assume that $U(0+)$ exists instead of the stronger $U(0+)=0$}. { In \cite{BFrobustez1}, on whose results we rely, it is assumed for simplicity that $U$ is unbounded from above, but this can be easily dispensed with from their work.} 
% can be assumed w.l.g.\ that  $U(0+)=0$.   Also,   under the latter condition we have $V\geq 0$.
 %More on this in Remark \ref{extra}.
\end{remark}

The usual way to dealing with the issue of existence of an element $\hat{X}\in \mathcal{X}(x)$  satisfying
$$\EE[U(\hat{X}_T)]\geq \EE[U({X}_T)]\:\hspace{0.3cm} \mbox{for all } \; X\in\mathcal{X}(x),$$  
uses crucially a result usually referred to as Kolmos Theorem. This result states that, from a sequence of random variables which is bounded in probability, one can extract a subsequence of convex combinations convergent in probability. To apply this, one also needs growth conditions on $U$ and $U'$ (see e.g.   \cite{KrSch} or \cite[Theorem 7.3.4]{Pham}). However, as a corollary of the analysis in \cite{BFrobustez1} the authors show in \cite[Proposition 5.22]{BFrobustez1}  that a shorter if more involved compactness argument can be applied; the same idea will allow us to prove the {sensitivity results for $u^{w}$} in the next section.  
 
The desired compactness property mentioned above holds in a suitably designed space. {In order to motivate it,} we start by observing that for $X\in \mathcal{X}(x)$:
$$\sup_{Y\in\Y}\EE\left [YU^{-1}\circ U(X) \right]\leq x,$$ 
where $\Y:=\Y_{\PP}(1)$, and we (now and often hereafter) write $Y$ for $Y_T$ and $X$ for $X_T$, as long as the context is unequivocal. We then see that setting
\begin{equation}
J(\cdot):=\sup_{Y\in\Y}\EE\left [YU^{-1}(|\cdot |) \right], \label{Jfinal}
\end{equation}
for every $X\in \mathcal{X}(x)$ we have that $J(U(X))\leq x$. {We remark that \eqref{Jinicial} and \eqref{Jfinal} coincide by \cite[Proposition 3.1]{KrSch}, so notation is consistent.} Therefore we may conjecture that if $J$ was connected to a norm (or say, grew stronger than it) and if the space defined by such a norm, which we shall soon call $L_J$, was a strong dual one, then we would get the weak* relative compactness of the set $\{U(X):X\in \mathcal{X}(x)\}$  {immediately from Banach-Alouglu's Theorem}.

Let us now summarize the main topological results in {\cite[Section 5]{BFrobustez1}} for future reference. {Consider $J$ as above and define $I: L^0 \to \RR \cup \{+\infty\}$ as}
$$
I(Z):= \inf_{Y\in\Y}\EE\left[\vert Z\vert V\left({Y}/{\vert Z\vert}\right)\right].%\hspace{0.8cm}\mbox{and} \hspace{0.8cm}
%J(X):= \sup_{Y\in \Y}\EE\left[YU^{-1}(\vert X\vert)\right],
$$
\begin{lemma}
Under Assumption {\bf({A1})}, the functions $I$ and $J$ are convex.
\end{lemma}
\begin{proof}
See {\cite[Lemma 5.1]{BFrobustez1}}.
\end{proof}\smallskip

We consider the spaces
\begin{align*}
L_I &:= \left\{Z\in L^0: I(\alpha Z)<\infty \mbox{ for some }\alpha>0\right\}, &
E_I := \left\{Z\in L^0: I(\alpha Z)<\infty \mbox{ for every }\alpha>0\right\},\\
L_J &:=\left\{Z\in L^0: J(\alpha Z)<\infty \mbox{ for some }\alpha>0\right\}, &
E_J := \left\{Z\in L^0: J(\alpha Z)<\infty \mbox{ for every }\alpha>0\right\},
\end{align*}
and for $F$ denoting $I$ or $J$, we set the equivalent norms (see \cite[Theorem 1.10]{Musielak}):
\begin{align}
\| s \|_{F,\ell} &:= \inf\{\beta>0: F(s/\beta)\leq 1 \} &\| s \|_{F,a} &:= \textstyle \inf\left\{ \frac{1}{k}+ \frac{F(ks)}{k}: k>0 \right\}.\label{normasellya}
\end{align}
\begin{lemma} Under Assumption {{\bf(A1)}} and after identifying almost equal elements,  for $\gamma=\ell,a$ we have that $(E_{F}, \| \cdot \|_{F,\gamma} )$, $(L_{F}, \| \cdot \|_{F,\gamma})$ are normed linear spaces. Moreover, $E_F$ is a closed subspace of $L_F$ and both  $E_I$ and $L_J$ are Banach spaces. 
\end{lemma}

Now, let us define $\Y^*:=\{Y\in\Y:Y>0 \mbox{ and }\forall \beta>0, \EE[V(\beta Y)]<\infty\}$ and suppose
$$
{{\bf(A2)}} \hspace{0.6cm} \Y^*\neq\emptyset, \hspace{3pt} I(Z)=\inf_{Y\in \Y^*}\EE[|Z|V(Y/|Z|)] \mbox{ , and }  J(X)=\sup_{Y\in \Y^*}\EE[YU^{-1}(|X|)].  
$$
\begin{remark} \label{atajo}
{Condition ${{\bf(A2)}}$ is satisfied for instance if the price process $S$ satisfies that $\dd S =\lambda \dd\langle M\rangle+\dd M$} for a continuous martingale $M$, ${\lambda \in L^2_{loc}(M)}$, the market model is viable and $\EE\left[V\left(\beta\E(\lambda\cdot M)_T\right) \right]<\infty$ for all $\beta >0$. See {\rm\cite[Lemma 5.7]{BFrobustez1}} for a proof of this fact.
\end{remark}
{The next result, proved in \cite[Proposition 5.10]{BFrobustez1},  establishes that $L_J$ is a strong dual space. }
\begin{theorem}
\label{teoresumen} Suppose that Assumptions {{\bf(A1)-(A2)}}  hold true. Then, the dual of $(E_I,\|\cdot\|_{I,a})$ is isometrically isomorphic to $(L_J,\|\cdot\|_{J,\ell})$.
%
%Then $I,J$ are convex and all the previous spaces are linear. Furthermore, if $F$ denotes either $I$ or $J$, then the following are equivalent norms on $L_F$ (after identifying almost equal elements):

%With the obvious identification, every element in $L_I$ (respect. $L_J$) can be identified with an element in the dual of $L_J$ (respect. $L_I$)
\end{theorem}

{

To wrap up, and in light of the expression \eqref{lamano} for $u^w$, we have given in this section conditions under which this weakly perturbed value function can indeed be viewed as a support function of a weakly compact set, namely $\{Z:J(Z)\leq x\}$. We  proceed in the next section to take advantage of this fact, in the context outlined in Section \ref{preliminares}, in order to perform the sensitivity analysis of our problem under weak perturbations.

\begin{remark} The spaces $L_J,L_I$ are examples of so-called modular spaces, which are generalizations of Orlicz spaces introduced by H. Nakano {\rm(}see \cite{Nakano2,Musielak}{\rm)}.
By e.g.\ H\"older inequality for modular spaces {\rm(}see {\rm \cite[Proposition 5.9]{BFrobustez1}}{\rm)} we have that $u^w$, given by \eqref{lamano}, is finite. Moreover, under our assumptions, {\rm\cite[Proposition 5.22]{BFrobustez1}} shows that the supremum therein is attained. Finally, it is easy to see that this optimizer is unique, as it must lie in the image set of $U$, which is a strictly concave function.
\end{remark}}

\section{Stability and sensitivity}
 \label{posssens}

Let us go back to the weakly  perturbed problem defined in \eqref{lamano} for some fixed  parameters $\bar{\mu} \in (L^{\infty,\infty}_{\FF})^{d}$ and $\bar{\sigma}  \in (L^{\infty,\infty}_{\FF})^{d\times n}$. We initially make the following assumption:\medskip\\
%\mu\sigma
{{\bf {(H2)}}} The utility function has the form  $U(x)= px^{1/p}$  ($p\in (1,+\infty)$) if $x\geq 0$ and it is equal to $-\infty$ otherwise. \smallskip

We shall first prove first Theorem \ref{teoremarkmenossimpli} under this assumption, namely:
\begin{theorem}\label{teoremarksimpli} Assume {\bf(H2)}. Consider some perturbations $(\Delta \mu, \Delta \sigma) \in  (L^{\infty, \infty}_{\FF})^{d}\times (L^{\infty, \infty}_{\FF})^{d\times n}$ and suppose that {\bf(H1)} is satisfied for  $(\mu^{\tau}, \sigma^{\tau}):=(\bar{\mu}+ \tau \Delta \mu, \bar{\sigma}+\tau \Delta\sigma)$  and small enough $\tau$. Then, the directional derivative $Du^{w}(\bar{\mu},\bar{\sigma})(\Delta \mu, \Delta \sigma)$ exists and is given by
%\footnotesize
%$$\mathcal{P}:=\left\{({\mu},{\sigma})\in (L^{\infty, \infty}_{\FF})^{d}\times (L^{\infty, \infty}_{\FF})^{d\times n} : \sigma\sigma^{\top} \mbox{ is a.e. invertible and } \esssup_{t,\omega} | [\sigma\sigma^{\top}]^{-1}|<\infty\right\}.$$ 
%\normalsize
%Fix $(\bar{\mu},\bar{\sigma})\in \mathcal{P}$ and define the function $u^w$ around it as in \eqref{lamano}. Then $u^{w}$ is continuous, G\^ateaux and directionally Hadamard differentiable at $(\bar{\mu}, \bar{\sigma})\in \mathcal{P}$. The directional derivatives are given by:
%
%
\begin{align*}\textstyle
%D_ru(\bar{r},\mu,\sigma)\Delta r &=& \EE \left [ U(\bar{X}(T))L \left\{ -\int_0^T [\sigma^{\top}\{\sigma\sigma^{\top}\}^{-1}\Delta r{\bf 1}]^{\top}\dd W + \int_0^T \langle \lambda-\bar{\lambda},\sigma^{\top}\{\sigma\sigma^{\top}\}^{-1}\Delta r{\bf 1} \rangle\dd t \right\} \right ] \\
D_{\mu}u^{{w}}(\bar{\mu},\bar{\sigma})\Delta \mu &=\textstyle \EE \left [ U(\bar{X}(T))  \int_0^T [\bar{\sigma}^{\top}[\bar{\sigma}\bar{\sigma}^{\top}]^{-1}\Delta \mu]^{\top}\dd W \right ], \\[4pt]
D_{\sigma}u^{{w}}(\bar{\mu},\bar{\sigma})\Delta \sigma &= \textstyle \EE \left [ U(\bar{X}(T))   \int_0^T \left\{\Delta\sigma^{\top}[\bar{\sigma}\bar{\sigma}^{\top}]^{-1}\bar{\mu}-  \bar{\sigma}^{\top} [\bar{\sigma}\bar{\sigma}^{\top}]^{-1}  [\bar{\sigma}\Delta\sigma^{\top} +\Delta\sigma\bar{\sigma}^{\top} ]  [\bar{\sigma}\bar{\sigma}^{\top}]^{-1}\bar{\mu}\right\}^{\top}\dd W  \right],
\end{align*}
where $\bar{X}(T)$ is the unique optimal terminal wealth attaining  $u(\bar{\mu},\bar{\sigma})$. Moreover, the application $(\mu, A) \in  (L^{\infty, \infty}_{\FF})^{d}\times (L^{\infty, \infty}_{\FF})^{d\times d}\mapsto  u^{w}(\mu, A (\bar{\sigma}\bar{\sigma}^{\top})^{-1}\bar{\sigma}) \in \RR$ is Hadamard differentiable at $(\bar{\mu},\bar{\sigma}\bar{\sigma}^{\top})$.
\end{theorem}

\smallskip
We denote by $q:=p/(p-1) \in (1,+\infty)$ the conjugate exponent of $p$.

\begin{remark}  \label{Ugral} {In the more general context of the previous section, we clearly  have that {\bf (H2)} implies  {\bf(A1)} and, thanks to Remark   \ref{atajo},  assumption {\bf(A2)} also holds true}.   %We stress that {\bf (H2)} could be modified to include more general utility functions {\color{red}(see Theorem \ref{teoremarkmenossimpli})}.
% Indeed all we really use in the arguments of the proofs to follow in this section is that $U$ grows slower than some $p$-root. An example of $U$ satisfying this and assumption  {\bf(A1)}, beyond the power case, is given e.g.\ by the inverse function of $x\mapsto e^x-x-1$. 
\end{remark}
% 
%%%%%%  
%\begin{remark}
{In the jargon of Section \ref{utilmax}, using the power-like form of the utility function we have that}
$$
\ba{rcl}
L_I &=&\textstyle \left\{Z\in L^0 : \inf_{\nu\in K(\bar{\sigma})}\EE \bigl [ \mathcal{E}(-\int [\bar{\lambda}+\nu]^{\top} \dd W)_T^{1-q} |Z|^q\bigr ]<\infty\right\},
%&=& \left\{Z\in L^0(\F_T): \inf_{\nu\in K(\bar{\sigma})}\EE \bigl [ |Z|^q\exp\left\{\int_0^T  [\bar{\lambda}+\nu]^{\top} \dd W + \frac{1}{2}\int_0^T  |\bar{\lambda}+\nu|^2 \dd t\right \} \right ]<\infty\right\},
\ea$$
%%%%
where
$$K(\bar{\sigma}):=\left\{\nu\in L^2_{loc}(W):\nu(t,\omega)\in {\mbox{Ker}}(\bar{\sigma}(t,\omega))\mbox{ a.e. } \right\},$$
%{\color{blue}since by the Kunita-Watanabe decomposition, \cite[Proposition 7.2.1]{Pham} and the fact that we are working in the brownian filtration, we know that every $Y\in\mathcal{Y}$ can be written as $Y_t=L_t\mathcal{E}(-\int [\bar{\lambda}+\nu]^{\top} \dd W)_t$ for some $ \nu \in K(\bar{\sigma})$ and {\color{red}$L$} decreasing, predictable and with $L_0=1$ (see also \cite[Proposition 3.2]{ZitLar}).}
as easily follows from \cite[Proposition 3.2]{ZitLar} and the fact that we are working on the Brownian filtration.
In this context, we have that $L_I=E_I$ and {for some constant $C(p)\in \RR_{++}$}
\be\label{normaI}\textstyle \|Z\|_I:=\|Z\|_{I,\ell}= C(p)\left (\inf_{\nu\in K(\bar{\sigma})}\EE \left [\mathcal{E}\left(-\int [\bar{\lambda}+\nu]^{\top} \dd W\right)_T^{1-q} |Z|^q \right ] \right )^\frac{1}{q}.\ee
{Analogously}, 
%%%%
\begin{eqnarray*}
L_J &=& \textstyle\bigl\{X\in L^0: \sup_{\nu\in K(\bar{\sigma})}\EE \left [\mathcal{E}\left(-\int [\bar{\lambda}+\nu]^{\top} \dd W\right)_T {|X|}^p \right ]<\infty\bigr\} ,
\end{eqnarray*}
we have that $L_J=E_J$ and {there exists a constant $c(p)\in \RR_{++}$ such that}
%%%
$$ \textstyle\|X\|_J:=\|X\|_{J,a}= c(p)\left (\sup_{\nu\in K(\bar{\sigma})}\EE \left [ \mathcal{E}\left(-\int [\bar{\lambda}+\nu]^{\top} \dd W\right)_T {|X|}^p\right ] \right )^\frac{1}{p}.$$ %Thus $L_J$ is just an $L^2(\QQ)$ space, where $\dd\Q:= \mathcal{E}(-\int \lambda^{\top} \dd W)_T\dd \PP$. 
Since $c(p)$ and $C(p)$ play no role here, we shall ignore them.  We state now a simple lemma that we shall invoke more than once:
\begin{lemma}
\label{integrabilidad} The following assertions hold true: \smallskip\\
{\rm(i)}  Let $\rho\geq 2$, $A\in (L_{\FF}^{\infty, \infty})^{n}$,  $B$ progressive, $n$-dimensional, such that $\EE\left[\int_0^T|B_t|^{\rho}\dd t\right]<\infty$ and $Z$ defined as the real-valued process solving $\dd Z= (ZA+B)^{\top}\dd W $. Then, there exists a constant  $c=c(\rho,T)>0$ such that 
$$\textstyle \EE\bigl[\sup_{s\in[0,T]}\vert Z_s\vert^{\rho}\bigr]\leq c\bigl[ \vert Z_0\vert^{\rho}+ \EE\bigl[\int_0^T|B_t|^{\rho}\dd t\bigr] \bigr]\exp\left\{ cT\|A\|^{\rho}_{\infty,\infty} \right\}. $$
{\rm(ii)} For every $\Gamma\in [L^{\infty,\infty}_{\FF}]^n$ we have $\E\left( \int \Gamma^{\top}\dd W\right)_T\in L_I$.
\end{lemma}
\begin{proof} {The proof of the first assertion is a standard application of Gronwall's Lemma (see e.g. \cite[Chapter 6, Section 4]{YongZhou})}. 
%
%
%We shall call $c$ generic constants depending on ${\rho},T$. For the first point, we know by BDG and Jensen's inequality that
%\begin{align*}
%\EE[\vert Z_t -Z_0\vert^{\rho}]&\leq c\EE\left [ \left (\int_0^t\vert Z(s)A(s)+B(s)\vert^2\dd s \right )^{{\rho}/2}\right ]\\
%&\leq c\left\{\EE\left[\int_0^t\vert Z(s)\vert^{\rho}\vert A(s)\vert^{\rho} \dd s \right] + \EE\left[\int_0^t\vert B(s)\vert^{\rho}\dd s \right]\right\},
%\end{align*}
%and so by Doob's inequality
%$$\EE\left[\sup_{s\in[0,t]}\vert Z_s\vert^{\rho}\right]\leq c\left \{ \vert Z_0\vert^{\rho}+ \|A\|^{\rho}_{\infty,\infty} \int_0^t \EE\left[\sup_{l\in [0,s]}\vert Z(l)\vert^{\rho} \right] \dd s +\EE\left[\int_0^t\vert B(s)\vert^{\rho}\dd s \right]  \right \},$$
%and so finally by Gronwall's inequality
%$$\EE\left[\sup_{s\in[0,T]}\vert Z_s\vert^{\rho}\right]\leq c\left \{ \vert Z_0\vert^{\rho}+\EE\left[\int_0^T\vert B(s)\vert^{\rho}\dd s \right]  \right \}\exp\left\{ cT\|A\|^{\rho}_{\infty,\infty} \right\}.$$
%%%
For the second point, using that $\bar{\lambda}$ and $\Gamma$ are essentially bounded,  we observe that $$ \textstyle\alpha:=\EE \bigl [ \E\left(\int \Gamma^{\top}\dd W\right)_T^q\exp\bigl\{\int_0^T  (q-1)\bar{\lambda}^{\top} \dd W + \frac{q-1}{2}\int_0^T  |\bar{\lambda}|^2 \dd t\bigr \} \bigr ],$$ satisfies
$$\textstyle\alpha\leq c \EE\left [\E\left(\int(q\Gamma+(q-1)\bar{\lambda})^{\top}\dd W\right)_T\right]=c,$$
for some constant $c>0$. Since $\alpha$ dominates $\|\E\left( \int \Gamma^{\top}\dd W\right)_T \|_I^q$, the result follows.
\end{proof}\smallskip

%{We denote $\bar{\lambda}:=\bar{\sigma}^{\top}[\bar{\sigma}\bar{\sigma}^{\top}]^{-1}\bar{\mu} \in [L^{\infty,\infty}_{\FF}]^{\color{blue}n}$}. 
%%By \cite[Theorem 4.2  and Remark 4.11]{KaraShreveFinance} we have that the market model is viable. 
%Now, let us define \small
%$$\mathcal{P}:=\left\{(\mu,\sigma)\in (L^{\infty, \infty}_{\FF})^{d}\times (L^{\infty, \infty}_{\FF})^{d\times n} : {(\sigma\sigma^{\top})^{-1}} \mbox{ {a.e. exists and }} \;  \esssup_{t,\omega} | [\sigma\sigma^{\top}]^{-1}|<\infty\right\}.$$ 
\normalsize
%{\color{blue}Vale la pena definir u?}For $(\mu,\sigma)\in \P$ and $\lambda= \sigma^{\top}[ \sigma \sigma^{\top}]^{-1}\mu$   we rename $u^w$ as
%\begin{equation}
%u(\lambda):= \sup_{\pi \in \Pi}\EE^{\PP^{\lambda}}\left[U(X_T^{\pi}) \right ],\label{probcentral}
%\end{equation}
%where $\dd \PP^{\lambda}:= \mathcal{E}(\int [\lambda-\bar{\lambda}]^{\top} \dd W )_{T}\dd\PP$.  {Novikov's condition implies that  $\PP^{\lambda}$ is indeed a probability measure}. 
Our aim now is to study the differentiability of $$\lambda \in \P \mapsto {u^{w}}(\lambda)\in \RR.$$  {First, let us define $g: (L_{\FF}^{\infty,\infty})^{n} \to L_I$ as}
$$\textstyle g(\lambda):= \mathcal{E}\left(\int [\lambda-\bar{\lambda}]^{\top} \dd W \right)_T.$$
Lemma \ref{integrabilidad}{\rm(ii)} implies that $g$ is well-defined.  We  prove now the Fr\'echet differentiability of $g$:
% , by repeatedly introducing SDEs and invoking Proposition \ref{qmmqemqmndndndndaaa}. We hope that a related reasoning can be helpful in the future for different utility functions, at least of power type. For the different notions of directional differentiability see \cite[Chapter 2.2.1]{BonSha}.

\begin{lemma}\label{lemg}
The map $g$ is  locally Lipschitz and  Fr\'echet differentiable. Moreover,  {for all $\Delta \lambda \in   (L_{\FF}^{\infty,\infty})^{n}$} we have that
\be\label{expresionderivadagateaux} \textstyle
Dg(\lambda)\Delta \lambda = \mathcal{E}\left(\int [\lambda-\bar{\lambda}]^{\top} \dd W \right)_T \left\{ \int_0^T \Delta\lambda_t^{\top}\dd W_t - \int_0^T   {(\lambda_t-\bar{\lambda}_t)\cdot\Delta \lambda_t}  \dd t \right\}.
\ee
\end{lemma}
\begin{proof} Let $\lambda_{1}$, $\lambda_2 \in (L_{\FF}^{\infty,\infty})^{n}$. We have that, omitting the dependence on $t$ and {denoting by $\|\cdot\|_2$ the $L^2$-norm with respect to $\PP$}, 
\be\label{locallipschitz}\textstyle
\| g(\lambda_1)- g(\lambda_2) \|_{I}^{q}\leq \left\| e^{\int_{0}^{T} (q-1)\bar{\lambda} \dd W + \frac{q-1}{2} \int_{0}^{T} |\bar{\lambda}|^2 \dd t}  \right\|_{2}\left\| |{(g(\lambda_1)-g(\lambda_2)})_T|^{q} \right\|_{2}.
\ee
{Note that   $\Delta g:=g(\lambda_1)-g(\lambda_2)$} solves
$$\dd \Delta g= \left[\Delta g (\lambda_1-\bar{\lambda})+ g(\lambda_2)(\lambda_1-\lambda_2)\right]^{\top} \dd W_t, \; \; t\in [0,T], \; \;  \Delta g_0=0,$$
and so the local Lipschitz property follows from Lemma \ref{integrabilidad}  and \eqref{locallipschitz}.  Let us prove that $g$ is G\^ateaux differentiable. Take $\lambda$ and call ${\bar{\Lambda}}=\lambda-\bar{\lambda}$ and $\lambda^{\epsilon}:=\bar{\Lambda} + \epsilon {\Delta \lambda}$. We see that
%e^{\int\bar{\Lambda}\dd W -\frac{1}{2}\int\vert \bar{\Lambda}\vert\dd t }
$$\textstyle \E\left(\int[\lambda^{\epsilon}]^{\top}\dd W\right)=\E\left(\int\bar{\Lambda}^{\top}\dd W\right)\exp\bigl\{\epsilon\int \Delta \lambda^{\top}\dd W-\epsilon \int\Delta \lambda \cdot\bar{\Lambda}\dd t -\frac{\epsilon^2}{2} \int \vert \Delta \lambda\vert^ 2\dd t\bigr\}.$$
Using that $e^x=1+x+x \int_0^1[e^{ax}-1] \dd a $ and calling $x_{\epsilon}$ the term inside $\exp\{\dots\}$ in the expression above, 
%$$x=\epsilon\int \Delta \lambda^{\top}\dd W-\epsilon \int\Delta \lambda\cdot\bar{\Lambda}\dd t -\frac{\epsilon^2}{2} \int \vert \Delta\vert^ 2\dd t$$
we obtain
$$\ba{rcl}
\frac{\E(\int[\lambda^{\epsilon}]^{\top}\dd W)-\E(\int\bar{\Lambda}^{\top}\dd W)}{\epsilon}&=& \E\left(\int\bar{\Lambda}^{\top}\dd W\right)\left[\int \Delta \lambda^{\top}\dd W- \int\Delta \lambda\cdot\bar{\Lambda}\dd t -\frac{\epsilon}{2} \int \vert \Delta \lambda\vert^ 2\dd t\right]\\[6pt]
\; & \; &+\epsilon^{-1}{x_\epsilon}\E\left(\int\bar{\Lambda}^{\top}\dd W\right)\int_0^1[e^{a{x_\epsilon}}-1]\dd a.
\ea$$
In order to show \eqref{expresionderivadagateaux}, it suffices to prove that $\|\E\left(\int\bar{\Lambda}^{\top}\dd W\right)\int \vert \Delta \lambda \vert^2\dd t\|_I<\infty$ and  
$$\textstyle\epsilon^{-1}\bigl\|{x_\epsilon}\E\left(\int\bar{\Lambda}^{\top}\dd W\right)\int_0^1[e^{a{x_\epsilon}}-1]\dd a\bigr\|_I\rightarrow 0  \hspace{0.3cm} \mbox{{as $\epsilon \to 0$}.}$$
The first claim is trivial, as $\Delta \lambda \in (L_{\FF}^{\infty,\infty})^{n}$ and $g({\lambda}) \in L_I$. For the second one, letting  $\nu\equiv 0$ {in \eqref{normaI}},   it suffices  to estimate  
$$\textstyle \EE\left [e^{\int (q-1)\bar{\lambda}^{\top}\dd W + \frac{(q-1)}{2}\int \vert \lambda \vert^2\dd t} \E\left(\int\bar{\Lambda}^{\top}\dd W\right)^q\left(\frac{{x_\epsilon}}{\epsilon}\right)^{q}  \bigl (\int_0^1[e^{a{x_\epsilon}}-1]\dd a\bigr )^q \right ],$$  
which we may bound from above by the product of
$$\textstyle\sqrt{\EE\bigl [\E\left(\int\bar{\Lambda}^{\top}\dd W\right)^{2q}\left[\int \Delta \lambda^{\top}\dd W-\int\Delta \lambda\cdot\bar{\Lambda}\dd t -\frac{\epsilon}{2} \int \vert \Delta \lambda\vert^ 2\dd t \right]^{2q}\bigr ]},$$
and 
$$\textstyle\sqrt{\EE\bigl [e^{2(q-1)\int \bar{\lambda}^{\top}\dd W + (q-1)\int \vert \bar{\lambda} \vert^2\dd t}\bigl (\int_0^1[e^{a{x_\epsilon}}-1]\dd a\bigr )^{2q} \bigr ]}.$$
Using the Cauchy-Schwartz and the {Burkholder-Davis-Gundy (BDG)} inequalities we have that the first term is finite. As for the second one, in order  to prove that it converges to zero it suffices to show that $\EE\left [ \int_0^1{\left|e^{a{x_\epsilon}}-1\right|}^{4q}\dd a \right]\rightarrow 0$. The term within the integral converges a.e.\ to zero as $\epsilon\to 0$. On the other hand, for some $c>0$, 
$$\textstyle |e^{ax_\epsilon}-1|^{4q}\leq c\bigl\{ 1 + e^{4q\vert \int \bar{\Lambda}\cdot\Delta \dd t\vert+4qa\epsilon \int\Delta^{\top}\dd W }\bigr\},$$
and $e^{4qa\epsilon \int\Delta^{\top}\dd W }\leq e^{4q \int\Delta^{\top}\dd W  }+ 1 $, which is integrable. {Thus, by dominated convergence, we have that \eqref{expresionderivadagateaux} holds true.}
%Therefore, we have proved the G\^ateaux differentiability of $g$ and established that $Dg(\lambda)(\Delta) = \E\left(\int\bar{\Lambda}^{\top}\dd W\right)_T\left[\int \Delta^{\top}\dd W- \int\Delta\cdot\bar{\Lambda}\dd t \right]$. 
	
{In order to prove Fr\'echet differentiability it suffices to show the continuity of the application  $\lambda \in  (L_{\FF}^{\infty,\infty})^{n} \mapsto Dg(\lambda)(\cdot)\in \mathcal{L}((L^{\infty,\infty}_{\FF})^n,L_I)$, where $\mathcal{L}((L^{\infty,\infty}_{\FF})^n,L_I)$ denotes the space of linear bounded operators from $(L^{\infty,\infty}_{\FF})^n$  to $L_I$}. Let $\Gamma$, $\lambda \in  (L^{\infty,\infty}_{\FF})^n$ and $\Delta \lambda\in (L^{\infty,\infty}_{\FF})^n$ such that  $\|\Delta \lambda\|_{\infty,\infty}\leq 1$. The triangle inequality yields 
\be\label{eqdegees}\ba{ll}\textstyle
\|[Dg(\lambda)-Dg(\Gamma)]\Delta \lambda \|_I \leq& \left\| \left[ \E\left (\int (\lambda-\bar{\lambda})^{\top}\dd W\right )-\E\left (\int (\Gamma-\bar{\lambda})^{\top}\dd W\right ) \right ]\int\Delta \lambda^{\top}\dd W \right\|_I\\[6pt]
 \; & \textstyle+\left\| \E\left (\int (\lambda-\bar{\lambda})^{\top}\dd W\right )\int\Delta \lambda\cdot (\lambda-\Gamma)\dd t \right\|_I\\[6pt]
\; & \textstyle+\left\|\left[ \E\left (\int (\lambda-\bar{\lambda})^{\top}\dd W\right )-\E\left (\int (\Gamma-\bar{\lambda})^{\top}\dd W\right ) \right ]\int\Delta \lambda\cdot (\Gamma-\bar{\lambda})\dd t \right\|_I.
\ea\ee\normalsize
Up to taking $q$-root, the first and the third r.h.s\ terms can be bounded above, through repeated Cauchy-Schwartz, by
$$\textstyle \sqrt{\EE\left [e^{ 4(q-1)\int \bar{\lambda}^{\top}\dd W + 2(q-1) \int \vert \bar{\lambda} \vert^2\dd t}\right ]}\sqrt[4]{\EE \left [{\left|\int \Delta \lambda^{\top}\dd W\right|^{4q}} \right]} \sqrt[4]{\EE \left [\left|Z_\Gamma-Z_\lambda\right|^{4q} \right]},$$
and
$$\textstyle \sqrt{\EE\left [e^{ 4(q-1)\int \bar{\lambda}^{\top}\dd W + 2(q-1)\int \vert \bar{\lambda} \vert^2\dd t}\right ]}\sqrt[4]{\EE \left [\left|\int \Delta \lambda \cdot[\Gamma-\bar{\lambda}]\dd t\right|^{4q} \right]} \sqrt[4]{\EE \left [\left|Z_\Gamma-Z_\lambda\right|^{4q} \right]},$$
where $Z_\Gamma:=\mathcal{E}\left(\int [\Gamma-\bar{\lambda}]^{\top} \dd W \right)_T$ and $Z_\lambda:=\mathcal{E}\left(\int [\lambda-\bar{\lambda}]^{\top} \dd W \right)_T$. As in the proof of the local Lipschitzianity of $g$, we get that the last term in both expressions above tends to zero. Therefore, the {BDG inequality implies} that the first and third terms in \eqref{eqdegees} tend to zero uniformly w.r.t. $\Delta \lambda$ satisfying that $\|\Delta \lambda\|_{\infty,\infty}\leq 1$. 
%
%, since by calling $Z$ the difference of the stochastic exponentials therein, one has
%%
%$$\dd Z(t)= Z(t)(\lambda-\bar{\lambda})(t)^{\top}\dd W + \E\left(\int(\Gamma-\bar{\lambda})^{\top}\dd W\right)_t[\lambda-\Gamma](t)^{\top}\dd W, $$
%%
%and so by Lemma \ref{integrabilidad} we have 
%%
%$$\|Z\|_{4q,\infty}^{4q}\leq c\EE\left[\int_0^T\left\vert\E\left(\int(\Gamma-\bar{\lambda})^{\top}\dd W\right)_t[\lambda-\Gamma](t) \right\vert^{4q}\dd t\right ] \exp\{c\|\lambda-\bar{\lambda}\|_{\infty,\infty}^{4q}\},$$
%which converges to zero as $\|\lambda-\Gamma\|_{\infty,\infty}\to 0$.
Finally, 
$$\textstyle \left\| \E\left (\int (\lambda-\bar{\lambda})^{\top}\dd W\right )\int\Delta \lambda\cdot (\lambda-\Gamma)\dd t \right\|_I\leq T\|\lambda-\Gamma\|_{\infty,\infty}\left\| \E\left (\int (\lambda-\bar{\lambda})^{\top}\dd W\right )\right\|_I.$$
The result follows. 
% As for the second term in the r.h.s\ of the inequality for the $Dg$'s, one may easily extract a term of type $\|\lambda-\Gamma\|_{\infty,\infty}$ from it and likewise conclude.
%
%We now prove that $g$ is locally Lipschitz. First we deduce that for every vicinity $V$ (say a ball) of $\Gamma$ and each $\Delta \in [L^{\infty,\infty}_{\FF}]^n$, the number $\sup_{\lambda\in V} \|Dg(\lambda)\Delta\|_I$ is finite (using triangle inequality and by the previous estimates). This proves by Banach-Steinhaus Theorem that in operator norm $c(V):=\sup_{\lambda\in V} \|Dg(\lambda)\|<\infty$. Now taking any $\lambda,\bar{\Lambda}\in V$:
%%
%\begin{equation}
%\|g(\bar{\Lambda})-g(\lambda)\|_I=\left\|\int_0^1Dg\left(\lambda+t[\bar{\Lambda}- \lambda]\right)(\bar{\Lambda}-\lambda)\dd t\right\|_I\leq \int_0^1\|Dg\left(\lambda+t[\bar{\Lambda}- \lambda]\right)(\bar{\Lambda}-\lambda)\|_I\dd t,
%\notag%\label{cotaLip}
%\end{equation}
%and so $\|g(\bar{\Lambda})-g(\lambda)\|_I\leq c(V)\|\lambda - \bar{\Lambda}\|$.
\end{proof}

{Using the above fact we prove the stability (continuity) and the Hadamard differentiability of $u^{w}$ as a function of the market price of risk $\lambda$. The reader is referred to the appendix for the definition of Hadamard directionally differentiable maps.  Some parts of the following proof are independent of the choice of utility function, pointing out that we may in the future extend our approach:}

\begin{proposition}\label{propopreli}
The function $u^{w}:(L^{\infty,\infty}_{\FF})^n\rightarrow \RR_+$ 
%{\color{blue}of Remark \ref{overload}} defined by
%%
%\be 
%u^{w}(\lambda)=\sup_{\pi \in \Pi} \EE\left[\mathcal{E}\left(\int [\lambda-\bar{\lambda}]^{\top} \dd W \right)_T U(X^{\pi}(T)) \right ],\label{efe}
%\ee
%%%
is continuous, {G\^ateaux} and  Hadamard directionally differentiable. Denoting by $X[\lambda]_T$ the optimal final wealth associated to $u^{w}(\lambda)$, which is unique,   for all $\Delta \lambda \in (L^{\infty,\infty}_{\FF})^n$ the directional derivative is given by 
\be\label{derivadadireccionalconellambda}\textstyle Du^w(\lambda)\Delta\lambda=\EE \left [  \mathcal{E}\left(\int [\lambda-\bar{\lambda}]^{\top} \dd W \right)_T U(X[\lambda]_T)\left\{ \int_0^T \Delta\lambda^{\top}\dd W - \int_0^T   (\lambda-\bar{\lambda})\cdot \Delta \lambda \dd t \right\} \right ].\ee
%

%and $Z(\lambda)$ denotes the unique element attaining $$\sup\left\{ \EE\left[\mathcal{E}\left(\int [\lambda-\bar{\lambda}]^{\top} \dd W \right)_T Z\right ]: Z\in\L_J^+, J(Z)\leq x \right\},$$
%%%
%or equivalently, $Z(\lambda)=U(X[\lambda]_T)$ and $X(\lambda)$ is the unique solution to \eqref{efe}.

\end{proposition}
\begin{proof} We have seen in \eqref{lamano} that  $u^{w}(\lambda)=\sup\left\{ \EE\left[g(\lambda) Z\right ]: Z\in L_J^+, J(Z)\leq x \right\}$. Define   $L_I \ni Y\mapsto F(Y):= \sup\left\{ \EE\left[Y Z\right ]: Z\in L_J^+, J(Z)\leq x \right\} \in \RR$, so that $u^w=F \circ g$. Theorem \ref{teoresumen} {and the  Banach-Alaoglu theorem}  imply  that the set $\{Z\in L_J^+: J(Z)\leq x\}$ is   weak* compact. Thus, Lemma \ref{lemmaenvelope}{\rm(ii)} in the appendix implies that $F$ is Hadamard directionally differentiable. So Lemma \ref{lemg}  and the chain rule in \cite[Theorem 2.28]{Penot} imply that $u^{w}$ is Hadamard directionally differentiable.
Its directional derivative is given by 
%%%%%%%%%555
$$\textstyle Du^{w}(\lambda)\Delta\lambda=\EE \left [ Z(\lambda)\mathcal{E}\left(\int [\lambda-\bar{\lambda}]^{\top} \dd W \right)_T \left\{ \int_0^T \Delta\lambda^{\top}\dd W - \int_0^T   (\lambda-\bar{\lambda})\cdot \Delta \lambda\dd t \right\} \right ],$$
with $Z(\lambda)=U(X[\lambda]_T)$. Using H\"older's inequality in \cite[Proposition 5.9]{BFrobustez1} we bound 
$$\textstyle |Du^{w}(\lambda)\Delta\lambda|\leq \left\|Z(\lambda)\right\|_J\left\| \mathcal{E}\left(\int [\lambda-\bar{\lambda}]^{\top} \dd W \right)_T \left\{ \int_0^T \Delta\lambda^{\top}\dd W - \int_0^T \langle \lambda-\bar{\lambda},\Delta \lambda \rangle\dd t \right\} \right\|_I.$$
Taking $k=1$ in \eqref{normasellya} and using that $J(Z(\lambda)) \leq x$, we obtain that  $\left\|Z(\lambda)\right\|_J\leq 1+x$. The second term in the expression above is {uniformly} bounded whenever $\Delta\lambda$ is taken in a bounded set (as in the proof in Lemma \ref{lemg}). Thus,  $Du^{w}(\lambda) (\cdot)$ is linear and continuous and so $ u^{w}$ is G\^ateaux differentiable.
\end{proof}

We can now prove Theorem \ref{teoremarksimpli}% \ref{preliminares}
\smallskip
\begin{proof}[Proof of Theorem \ref{teoremarksimpli}]
By {\bf (H1)} we have that $(\bar{\sigma}\bar{\sigma}^{\top})^{-1}$ is essentially bounded. Thus,   
{$$(\mu,\sigma)\in (L^{\infty,\infty}_{\FF})^{d} \times (L^{\infty,\infty}_{\FF})^{d\times n} \mapsto \lambda(\mu, \sigma):= \sigma^{\top}[\sigma\sigma^{\top}]^{-1} \mu,$$}
is Fr\'echet differentiable at $(\bar{\mu},\bar{\sigma})$ and its directional derivative is given by
\begin{eqnarray*}\textstyle
D\lambda(\mu,\sigma) (\Delta\mu,\Delta\sigma)&=&\textstyle \sigma^{\top}[\sigma\sigma^{\top}]^{-1} \Delta \mu +
 \Delta\sigma^{\top}[\sigma\sigma^{\top}]^{-1}\mu \\
&& \textstyle -  \sigma^{\top} [\sigma\sigma^{\top}]^{-1}  \left\{\sigma\Delta\sigma^{\top} +\Delta\sigma\sigma^{\top} \right\}  [\sigma\sigma^{\top}]^{-1}\mu.
\end{eqnarray*}
The result easily follows from Porposition \ref{propopreli}  and the chain rule in \cite[Theorem 2.28]{Penot}.
%Finally we note that $u(\bar{r},\mu,\sigma)=f\circ H(\bar{r},\mu,\sigma)$, for $f$ as in \eqref{efe}. By Proposition \ref{propopreli} we see that $u$ is indeed continuous and the composition of a Hadamard and a directionally differentiable function. By \cite[Proposition 2.47]{BonSha} we conclude that $u$ is directionally differentiable, and these derivatives are as stated.
%
%
%{\color{red}REVISAR:
%
% Finally, as we did at the end of Proposition \ref{propopreli}, we may use H\"older's inequality to prove the continuity of the directional derivatives proving the G\^ateaux differentiability.}
\end{proof}

\medskip

%
%

%\begin{remark}
%The extension of this result to more general utility functions and incomplete markets is an ongoing work.
%\end{remark}
%
%\begin{remark}\label{remarksimpli} If the market is complete, then $n=d$ and $\bar{\sigma}$ is invertible ({\color{blue}see e.g.\ \cite[Theorem 6.6, Chapter 1]{KaraShreveFinance}}). Therefore, in this  context, we get 
%$$Du^{w}\left((\bar{\mu},\bar{\sigma}), (\Delta \mu, \Delta \sigma) \right)= \EE \left [ U(\bar{X}(T))  \int_0^T [\bar{\sigma}^{-1}  \Delta \mu- \bar{\sigma}^{-1}\Delta \sigma \bar{\sigma}^{-1}\bar{\mu}]^{\top}\dd W \right ].$$
%\end{remark}
%\subsection{Some simple extensions of Theorem \ref{teofull}}

We now lift Assumption {\bf (H2)} and prove our main result Theorem  \ref{teoremarkmenossimpli}:

\begin{proof}[Proof of Theorem \ref{teoremarkmenossimpli} ]
We let $\bar{U}(x)=Cx^{1/p}$ and $\bar{V}$ its conjugate. Then for some other constant $c$ we have $V(y)\leq \bar{V}(y)= cy^{1/(1-p)}$ and so $zV(y/z)\leq cy^{1/(1-p)} z^{p/(p-1)}$. Writing $L_I$ for the modular space associated with $zV(y/z)$ and $L_{\bar{I}}$ for the one associated with $cy^{1/(1-p)} z^{p/(p-1)}$ (as it has been described throughout most of this section) we conclude that $L_{\bar{I}}\subset L_I$ with continuous injection. Let $i:L_{\bar{I}}\to L_I$ be the identity map, which is then linear continuous and thus Fr\'echet differentiable with $Di=i$. In particular $G:(L^{\infty,\infty}_{\FF})^n\to L_I$ given by  $ G(\lambda):= \mathcal{E}\left(\int [\lambda-\bar{\lambda}]^{\top} \dd W \right)_T$ is well defined, and we have $G=i\circ g$ with $g$ as before. By Lemma \ref{lemg} we conclude that $G$ is loc.\ Lipschitz and Fr\'echet differentiable with the same derivative as in \eqref{expresionderivadagateaux}. One can then argue as in Proposition \ref{propopreli} and the proof of Theorem \ref{teoremarkmenossimpli} to conclude.
\end{proof}

\begin{remark}
Note that the proof provides the Hadamard differentiability for the natural extension of $u^{w}$ to  {$ \mathcal{P}$, where $\mathcal{P}$ is defined as 
$$\textstyle \mathcal{P}:=\left\{(\mu,\sigma)\in (L^{\infty, \infty}_{\FF})^{d}\times (L^{\infty, \infty}_{\FF})^{d\times n} : \sigma\sigma^{\top} \mbox{ is a.e. invertible and } \esssup_{t,\omega} | [\sigma\sigma^{\top}]^{-1}|<\infty\right\},$$}
i.e.\ for perturbations not necessarily satisfying  the stability of the Kernels in {\bf(H1)}. However, this extension of $u^{w}$ for perturbations not satisfying  {\bf(H1)} is meaningless, as we have already discussed. 
\end{remark}

We now provide a one-sided second order bound for the first order approximation error. {It seems that a full second-order expansion or better, a sensitivity analysis of the optimal wealth, is beyond what we can reach by only looking at the primal problem. See \cite{prepLarsen} for such results via the duality method and for strong perturbations in the negative-power utility case.} For simplicity we only consider perturbations of the market price of risk around the reference parameter $\bar{\lambda}$.

\begin{proposition}
For any $\delta >0$ and $|\epsilon|\leq \delta$ we have
\begin{align}\textstyle \label{second order}
u^w(\bar{\lambda}+\epsilon \Delta\lambda) - u^w(\bar{\lambda})-\epsilon Du^w(\bar{\lambda})\Delta\lambda  \geq -C(\delta)\epsilon^2,
\end{align}
where $C(\delta)\geq 0$ and $Du^w$ is given by \eqref{derivadadireccionalconellambda}.

\end{proposition}

\begin{proof}
Denoting $\bar{Z}$ the optimizer for $u^w(\bar{\lambda})$, we have by H\"older's inequality
\begin{align*}\textstyle 
u^w(\bar{\lambda}+\epsilon \Delta\lambda) - u^w(\bar{\lambda})-\epsilon Du^w(\bar{\lambda})\Delta\lambda  & \geq \textstyle  \EE \bigl[ \bar{Z}\bigl \{ g(\bar{\lambda}+\epsilon \Delta\lambda)-1-\epsilon \int_0^T\Delta\lambda^{\top}\dd W  \bigr\}  \bigr] \\ 
& \geq \textstyle -\|\bar{Z}\|_{J}\bigl \| g(\bar{\lambda}+\epsilon \Delta\lambda)-1-\epsilon \int_0^T\Delta\lambda^{\top}\dd W   \bigr\|_I.
\end{align*}
Defining $Y_t:=g(\bar{\lambda}+\epsilon \Delta\lambda)_t-1-\epsilon \int_0^t\Delta\lambda^{\top}\dd W  $, we argue as in the proof of Lemma \ref{lemg} that
%%
%$$\textstyle dY_t = \bigl[Y_t(\bar{\lambda}^{\top}+\epsilon \Delta\lambda_t^{\top})+ \epsilon \bar{\lambda}_t^{\top} \int_0^t\Delta\lambda^{\top}_s\dd W_s  +\epsilon^2 \Delta\lambda_t^{\top} \int_0^t\Delta\lambda^{\top}_s\dd W_s  \bigr  ]\dd W_t, $$
%%
%so by SDE estimate in Lemma \ref{integrabilidad}, we find
%
%$$\textstyle \EE[Y_T^{q}]\leq c e^{k[\epsilon^{q} \|\Delta\lambda\|^{q}_{\infty,\infty} + \|\bar{\lambda}\|^{q}_{\infty,\infty}  ]} \EE \bigl [\int_0^T \bigl\{  |\bar{\lambda}_t|^q+\epsilon^q |\bar{\lambda}_t |^{q}(\int_0^t  \Delta\lambda^{\top}_s\dd W_s  )^{q}+ \epsilon^{2q}  |\Delta\lambda_t|^{q} (\int_0^t  \Delta\lambda^{\top}_s\dd W_s  )^{q}\bigr\} \dd t\bigr ],  $$
%%
%thus $\|Y_T\|_I^q\leq \tilde{C}(\delta)\epsilon^{2q} $ and we conclude.
%
$$\textstyle dY_t = \bigl[Y_t\epsilon \Delta\lambda_t^{\top} +\epsilon^2 \Delta\lambda_t^{\top} \int_0^t\Delta\lambda^{\top}_s\dd W_s  \bigr  ]\dd W_t, $$
so by  the SDE estimate in Lemma \ref{integrabilidad}, we find
$$\textstyle \EE[Y_T^{q}]\leq c e^{k\epsilon^{q} \|\Delta\lambda\|^{q}_{\infty,\infty} } \EE \bigl [\int_0^T \bigl\{  \epsilon^{2q}  |\Delta\lambda_t|^{q} (\int_0^t  \Delta\lambda^{\top}_s\dd W_s  )^{q}\bigr\} \dd t\bigr ],  $$
thus $\|Y_T\|_I^q\leq \tilde{C}(\delta)\epsilon^{2q} $ and we conclude.
\end{proof}

To conclude this section, we show how the results in Theorem  \ref{teoremarkmenossimpli} extend to the case of non trivial interest rate. 
%{\color{blue}
%
%In this subsection we briefly present how Theorem \ref{teofull} can be extended to perturbations of more model parameters. We consider the case of interest rates and the initial wealth. 
%
%\subsubsection{The case of non-trivial interest rate}
%}
%We can easily extend our results when a riskless bond is also considered in the market (i.e.\ prices are not discounted).
 More precisely, suppose now that the market comprises  the previous $d$ risky  assets $S^1, \hdots, S^d$ and also a riskless asset  $S^0$, satisfying that $\dd S^0_t= r_tS^0_t \dd t$, $S^0_0=s_0^{0}\in \RR_{++}$, with $r\in L^{\infty,\infty}_{\FF}$.  In this case the wealth process satisfies the SDE
$$\ba{rcl}
\dd X_t^{\pi}&=& \left[ r(t)X_t^{\pi} +\pi_t^{\top}(\mu_t-r_t{\bf 1}) \right] \dd t+\pi_t^{\top} \sigma_t \dd W_t, \hspace{0.4cm} \; t \in [0,T],\\[4pt]
X_0^{\pi}&=&x, \ea
$$
where ${\bf 1}$ denotes the vector of ones in $\RR^{d}$. Let us fix $(\bar{r},\bar{\mu},\bar{\sigma})\in L^{\infty,\infty}_{\FF}\times \P$ and for any $(r^{\tau},\mu^{\tau},\sigma^{\tau})\in L^{\infty,\infty}_{\FF}\times \P$ denote by $u^{s}(r^{\tau}, \mu^{\tau}, \sigma^{\tau})$ the value of the strongly perturbed problem. Then, by a simple change of variable, for a $p$-power utility function ($p\in (1,\infty)$) we find that 
$$u^{s}(r^{\tau}, \mu^{\tau}, \sigma^{\tau})= \sup_{   \pi \in \Pi } \EE^{\PP}\left( e^{\frac{1}{p}\int_{0}^{T} r^{\tau}_t \dd t}U(\hat{X}_{T}^{ \pi,\tau})\right),$$
where $\hat{X}_{T}^{ \pi,\tau}$ solves
$$\ba{rcl}
\dd \hat{X}_t^{\pi,{\tau}}&=&  \pi_t^{\top}(\mu^{\tau}_t-r^{\tau}_t{\bf 1})  \dd t+\pi_t^{\top} \sigma^{\tau}_t \dd W_t, \hspace{0.4cm} \; t \in [0,T],\\[4pt]
\hat{X}_0^{\pi,\tau}&=&x. \ea
$$
{Assuming that $\bar{\sigma}$ and $\sigma^\tau$ satisfy {\bf(H1)},} we then define the weakly perturbed value function as
$$u^{w}(r^{\tau}, \mu^{\tau}, \sigma^{\tau})= \sup_{ { \pi \in \Pi }} \EE^{\PP^{\tau}}\left( e^{\frac{1}{p}\int_{0}^{T} r^{\tau}_t \dd t}U(\hat{X}_{T}^{ \pi})\right),$$
with $\hat{X}_{T}^{ \pi}$ solving
$$\ba{rcl}
\dd \hat{X}_t^{\pi}&=&  \pi_t^{\top}(\bar{\mu}_t-\bar{r}_t{\bf 1})  \dd t+\pi_t^{\top} \bar{\sigma}_t \dd W_t, \hspace{0.4cm} \; t \in [0,T],\\[4pt]
\hat{X}_0^{\pi}&=&x, \ea
$$
and {$\dd\PP^{\tau}=\E\left[\int\left( \lambda^{\tau}_r - \bar{\lambda}_r\right)\dd W \right]_T\dd\PP$, with $\lambda^{\tau}_r := (\sigma^{\tau})^{\top}[\sigma^{\tau}(\sigma^{\tau})^{\top}]^{-1}(\mu^{\tau}-r^{\tau}{\bf 1})$ and  $\bar{\lambda}_r := \bar{\sigma}^{\top}[\bar{\sigma}\bar{\sigma}^{\top}]^{-1}(\bar{\mu}-\bar{r}{\bf 1})$.
}
Thus, arguing exactly as before we obtain the following sensitivities; for every $(\Delta r, \Delta \mu, \Delta \sigma) \in L^{\infty,\infty}_{\FF}\times (L^{\infty,\infty}_{\FF})^{d} \times (L^{\infty,\infty}_{\FF})^{d\times n}$ {such that $\sigma^{\tau}:= \bar{\sigma}+ \tau \Delta \sigma$ satisfies {\bf(H1)} for $\tau>0$ small enough}, we have
$$
\ba{ll} \textstyle
%D_ru(\bar{r},\mu,\sigma)\Delta r &=& \textstyle\EE \left [ U(\bar{X}(T))L \left\{ -\int_0^T [\sigma^{\top}\{\sigma\sigma^{\top}\}^{-1}\Delta r{\bf 1}]^{\top}\dd W + \int_0^T \langle \lambda-\bar{\lambda},\sigma^{\top}\{\sigma\sigma^{\top}\}^{-1}\Delta r{\bf 1} \rangle\dd t \right\} \right ] \\ \textstyle
D_{r}u^w(\bar{r},\bar{\mu},\bar{\sigma})\Delta r =& \textstyle \EE \left [e^{\frac{1}{p}\int_{0}^{T} r_t \dd t}U(\hat{X}_{T}^{ \pi})\left\{ \frac{1}{p}\int_{0}^{T} \Delta r_t \dd t -   \int_0^T [\bar{\sigma}^{\top}[\bar{\sigma}\bar{\sigma}^{\top}]^{-1}\Delta r {\bf 1}]^{\top}\dd W\right\} \right], \\[6pt] \textstyle
D_{\mu}u^w(\bar{r},\bar{\mu},\bar{\sigma})\Delta \mu =& \textstyle \EE \left [ e^{\frac{1}{p}\int_{0}^{T} r_t \dd t}U(\hat{X}_{T}^{ \pi})  \int_0^T [\bar{\sigma}^{\top}[\bar{\sigma}\bar{\sigma}^{\top}]^{-1}\Delta \mu]^{\top}\dd W \right ], \\[6pt] \textstyle
D_{\sigma}u^w(\bar{r},\bar{\mu},\bar{\sigma})\Delta \sigma =& \textstyle\EE \left [e^{\frac{1}{p}\int_{0}^{T} r_t \dd t}U(\hat{X}_{T}^{ \pi})   \int_0^T  \left[\Delta\sigma^{\top}[\bar{\sigma}\bar{\sigma}^{\top}]^{-1}(\bar{\mu}-\bar{r}{\bf 1})\right]^{\top} \dd W\right] \\[6pt] \textstyle
 &- \textstyle\EE \left [e^{\frac{1}{p}\int_{0}^{T} r_t \dd t}U(\hat{X}_{T}^{ \pi}) \int_{0}^{T} \left[\bar{\sigma}^{\top} [\bar{\sigma}\bar{\sigma}^{\top}]^{-1}  [\bar{\sigma}\Delta\sigma^{\top} +\Delta\sigma\bar{\sigma}^{\top} ]  [\bar{\sigma}\bar{\sigma}^{\top}]^{-1}(\bar{\mu}-\bar{r}{\bf 1})\right]^{\top}\dd W  \right].  \ea$$%
\section{A final discussion} 
\label{onvaluefunctions}

As we have seen in Section \ref{subsec counter} the sensitivities in the weak and strong formulations may differ. Proposition \ref{igualdadfuerteydebil} and {Remark \ref{difficulty_stochastic_parameters}} thereafter, on the other hand, give a hint as to why this happens. We close the article by providing an expression, which we derive  heuristically, connecting the sensitivities of the weakly and strongly perturbed problems. For simplicity, we restrict the analysis to varying market prices of risk only (and fixed volatilities, so only the drift is being perturbed). We work in canonical continuous-paths space.

Let us denote $\theta^{\epsilon}(\omega)=\omega + \epsilon \int \delta\lambda \dd s$ a shift in canonical space and $X^*$ the optimal wealth ($\pi^*$ the optimal portfolio) under reference parameters. Then
\begin{multline*}
\textstyle
\EE\left [ U(X^*(T)\circ \theta^{\epsilon}) \right ] - \EE\left [ U(X^*(T)) \right ]   = \\ \textstyle
 \EE\bigl [ U\bigl (  x + \int_0^T[\pi^*\cdot\bar{\lambda}] \circ \theta^{\epsilon}\dd s +\int_0^T \pi^*\circ \theta^{\epsilon}\cdot\dd W + \epsilon\int_0^T[\pi^* \circ \theta^{\epsilon}]\cdot \delta\lambda\dd s \bigr) \bigr ] - \EE\left [ U(X^*(T)) \right ]. 
 \end{multline*}
From this we conclude that, if the corresponding directional derivatives in path-space are well-defined, 
$$\textstyle\left .\frac{\dd}{\dd \epsilon}\EE\left [ U(X^*(T)\circ \theta^{\epsilon}) \right ]\right\vert_{\epsilon=0} 
=\EE\bigl [ U'(X^*(T)) \bigl \{ \int_0^T D[\pi_s^*\cdot\bar{\lambda}_s](\omega, \delta\lambda)\dd s +\int_0^T D\pi_s^*(\omega,\delta\lambda)\dd W_s + \int_0^T\pi^*\cdot\delta\lambda \dd s \bigr\}\bigr ].$$ 
Now, by Bismut's integration by parts formula (see e.g. \cite[Chapter IV, Section 41]{MR1780932} and the assumptions therein), under given conditions this implies: 
\begin{multline}\label{eqluegoderogers}\textstyle\EE\bigl[U(X^*(T))\int_0^T\delta\lambda^{\top}\dd W \bigr]= \\ \textstyle
\EE\bigl[ U'(X^*(T)) \bigl \{ \int_0^T D[\pi^*_s\cdot\bar{\lambda}](\omega, \delta\lambda)\dd s +\int_0^T D\pi_s^*(\omega,\delta\lambda)\dd W_s + \int_0^T\pi^*\cdot\delta\lambda \dd s \bigr\}\bigr ].
\end{multline}
We can reasonably conjecture, if anything like the ``envelope'' or ``Danskin Theorem'' is to hold for it, as well as a directional chain rule, that  
$$ \textstyle Du^s(\bar{\lambda})\delta\lambda =\EE\bigl[U'(X^*(T))\int_0^T \pi^*\cdot \delta\lambda \dd t\bigr],$$
in accordance to \cite{prepLarsen} for the case of negative power utility,
and so the l.h.s.\ in \eqref{eqluegoderogers} is the sensitivity associated to weak perturbations (see \eqref{derivadadireccionalconellambda}, evaluated at $\bar{\lambda}$) whereas the sensitivity for strong perturbations is contained in the r.h.s. Thus, we obtain the sought after relationship between sensitivities: 
\begin{equation}\label{eq heuristic}
 Du^w(\bar{\lambda})\delta\lambda-Du^s(\bar{\lambda})\delta\lambda= \EE\left [ U'(X^*(T)) \left \{ \int_0^T D[\pi_s^*\cdot\bar{\lambda}_s](\omega, \delta\lambda)\dd s +\int_0^T D\pi_s^*(\omega,\delta\lambda)\dd W \right\}\right ].
\end{equation}

\noindent It seems to us that a rigorous derivation of \eqref{eq heuristic} is an interesting, and challenging, open problem.

We now make use of \eqref{eq heuristic} to recover the result in Proposition \ref{igualdadfuerteydebil}. Let us assume that $\bar{\lambda}$ is deterministic and see what this can imply. Call $$ \textstyle R_t:=\int_0^t D[\pi^*_s](\omega, \delta\lambda)\cdot\bar{\lambda}\dd s +\int_0^t D\pi_s^*(\omega,\delta\lambda)\dd W_s =\int_0^t D[\pi^*_s](\omega, \delta\lambda)\cdot\{\bar{\lambda}\dd s +\dd W_s\}. $$
By duality and \cite[Corollary 3.3]{ZitLar} we know that  there is a scalar $a$ (making sure that $X^*(T)$ satisfies the budget constraint) such that $U'(X^*(T))=a\E\left(-\int[\bar{\lambda}+\nu] \dd W \right)$, for some $\nu\in K(\bar{\sigma})$; see Section \ref{posssens}. We then see by the product formula that, upon defining $\dd Z_t=-Z_t [\bar{\lambda}+\nu]\dd W_t$, we get:
\begin{multline*}\textstyle
\EE\bigl [ U'(X^*(T)) \bigl \{ \int_0^T D[\pi^*\cdot\bar{\lambda}] \delta\lambda\dd s +\int_0^T D\pi^*\delta\lambda\dd W \bigr\}\bigr ]=\textstyle a \EE[Z_TR_T]\\
=a \textstyle\EE\bigl [ \int_0^T R_t\dd Z_t\bigr ]+a\EE\bigl [ \int_0^T Z_tD[\pi^*_t](\omega, \delta\lambda)\cdot\{\bar{\lambda}\dd t+\dd W_t\}\bigr ]-a\EE\bigl[\int_0^T Z_t[\bar{\lambda}_t+\nu_t]\cdot D[\pi^*_s](\omega, \delta\lambda)\dd s\bigr ] .
\end{multline*}
Under enough integrability conditions 
%(e.g.\ if $\nu$ and $D[\pi^*_t](\omega, \delta\lambda)$ are essentially bounded)  {\color{red}(aqui elimine el parentesis pq quizas es mucho pedir)}
so that the Brownian integrals are martingales, we conclude \small
\begin{align*} \textstyle
\EE\bigl [ U'(X^*(T)) \bigl \{ \int_0^T D[\pi^*\cdot\bar{\lambda}] \delta\lambda\dd s +\int_0^T D\pi^*\delta\lambda\dd W \bigr\}\bigr ]& = \textstyle -a\EE\bigl[\int_0^T Z_t \nu_t\cdot D[\pi^*_s](\omega, \delta\lambda)\dd s\bigr ] ,
\end{align*} \normalsize
and recalling that an optimal $n$-dimensional $\pi^*$ corresponds to a $\bar{\sigma}^{\top}\pi$ in the original $d$-assets, we see that if $\bar{\sigma}$ is deterministic then the r.h.s.\ also vanishes. All in all, we obtain 
\be\label{igualdadsensibilidadesd} Du^w(\bar{\lambda})\delta\lambda=Du^s(\bar{\lambda})\delta\lambda,\ee
which is in tandem with our Proposition \ref{igualdadfuerteydebil}, as well as \cite[Lemma 9.2]{Davis06} and \cite[Theorem 3.1]{Monoyios} for instance.
\section*{Appendix}
{We provide the proof of a version of the envelope or Danskin's theorem (see \cite{MR0228260}), adapted to our purposes. }First, we recall the notion of Hadamard differentiability. Given two Banach spaces $(\X, \| \cdot \|_{\X})$ and $(\Z, \| \cdot \|_{\Z})$ a map $f: \X \to \Z$ is  directionally differentiable at $x$  if  for all $h\in \X$ the limit in $\Z$
$$ \textstyle Df(x,h):= \lim_{\tau \downarrow 0} \frac{f(x+\tau h)-f(x)}{\tau}, $$
exists. If in addition, for all $h\in \X$ the following equality in $\Z$ holds
$$\textstyle Df(x,h)= \lim_{\tau \downarrow 0, \; h'\to h} \frac{f(x+\tau h')-f(x)}{\tau} ,$$
then we say that $f$ is  directionally differentiable at $x$ in the Hadamard sense. 
An important property of Hadamard differentiable functions is the chain rule. More precisely, if $(\V, \| \cdot\|_{\V})$ is another Banach space,  $g: \V \to \X$ is directionally differentiable at $v$ and $f$ is directionally differentiable at $g(v)$ in the Hadamard sense, then the composition $f\circ g$ is directionally differentiable at $v$ (see e.g. \cite[Proposition 2.47]{BonSha}) and $D(f\circ g)(v, v')= Df(g(v),Dg(v,v'))$ for all $v'\in \V$. If in addition, $g$ is is also Hadamard directionally differentiable at $v$, then $f\circ g$ is  directionally differentiable at $v$ in the Hadamard sense.

\label{soporte}
Now, suppose that $K\subseteq \X$ is a weakly compact set. Let us consider the problem:
 $$\textstyle\sup_{Z\in X} \langle d, Z \rangle \hspace{0.3cm} \mbox{s.t. } \; Z\in K, \eqno(AP_{d})$$
where $d\in \X^{*}$ and $\langle \cdot , \cdot\rangle$ denotes the bilinear pairing between $\X$ and $\X^{*}$.  Let us define $v: \X^{*} \to \RR$ as the optimal value of problem $(AP_{d})$ and $\SS(d)$ the set of optimal solutions of  $(AP_{d})$, i.e.
$$\textstyle v(d):= \sup_{Z\in K} \langle d, Z \rangle,  \hspace{0.6cm} \SS(d):= \left\{ Z\in K \; ; \; v(d)= \langle d, Z \rangle\right\}.$$
Note that $v$ is well defined, it is a Lipschitz function  and {$\SS(d)\neq \emptyset$}. In fact, 
\be\label{emqmemqeqa}\textstyle | v(d_1)- v(d_2) | \leq \| d_1- d_2\|_{\X^{*}}\sup_{Z\in K} \| Z\|_{\X}.\ee
{The proof of the following result is  a simple modification of the proof in \cite[Theorem 4.13]{BonSha}.}
 \begin{lemma}\label{lemmaenvelope} For any $\bar{d}\in \X^{*}$, the following assertions hold true\smallskip\\
{\rm(i)} The set $\SS(\bar{d})$ is weakly compact.\smallskip\\
{\rm(ii)} The function $v$ is   directionally differentiable in the Hadamard sense  and   its directional derivative is
\be\label{mmmwnrnwr} \textstyle Dv(\bar{d}, \Delta d)= \sup_{Z\in \SS(\bar{d})} \langle \Delta d, Z \rangle \hspace{0.4cm} \mbox{{\rm for all} $\Delta d\in \X^{*}$. }  \ee
 \end{lemma}
 \begin{proof}  The first assertion follows directly from the weak-continuity of $\langle \bar{d}, \cdot\rangle$, which implies the weak closedness of   $\SS(\bar{d})$.  Now, in view of \cite[Proposition 2.49]{BonSha} {and \eqref{emqmemqeqa}} it suffices to show that $v$ is {directionally} differentiable.  Let $\bar{Z}\in S(\bar{d})$ be such that $ \langle \Delta d, \bar{Z} \rangle= \sup_{Z\in \SS(\bar{d})} \langle \Delta d, Z \rangle$  and for  $\tau>0$ set $d_{\tau}:=  \bar{d}+ \tau \Delta d$. By definition
$$ v(d_{\tau})- v(\bar{d})  \geq \langle d_{\tau}-\bar{d}, \bar{Z}\rangle= \tau \langle \Delta d, \bar{Z} \rangle, $$
which implies that 
\be\label{qmmenqnedasasad} \textstyle
\liminf_{\tau\to 0} \frac{ v(d_{\tau})- v(\bar{d})}{\tau} \geq  \langle \Delta d, \bar{Z} \rangle= \sup_{Z\in \SS(\bar{d})} \langle \Delta d, Z \rangle.
\ee
Analogously, let $Z_{\tau} \in S(d_{\tau})$. Then 
\be\label{qeqekqkednnddas}\textstyle v(\bar{d})- v(d_\tau)  \geq - \langle d_{\tau}-\bar{d}, Z_{\tau}\rangle= - \tau\langle \Delta d, Z_\tau \rangle.\ee
On the other hand, using \eqref{emqmemqeqa}  we get that $v(d_{\tau})\to v(\bar{d})$ as $\tau \downarrow 0$, which implies, since $d_{\tau}\to \bar{d}$ strongly in $\X^*$, that any weak limit point of $Z_{\tau}$ belongs to $\SS(\bar{d})$. Thus,  \eqref{qeqekqkednnddas} yields  
\be\label{amadnandad}\textstyle \limsup_{\tau\to 0} \frac{ v(d_{\tau})- v(\bar{d})}{\tau} \leq   \limsup_{\tau\to 0}  \langle \Delta d, Z_{\tau} \rangle \leq  \sup_{Z\in \SS(\bar{d})} \langle \Delta d, Z \rangle.\ee
Therefore,  \eqref{mmmwnrnwr}  is a consequence of  \eqref{qmmenqnedasasad} and \eqref{amadnandad}.
 \end{proof}
% \begin{remark}{\color{red} Enunciar el teorema con el resultado Frechet y decir sus consecuencias as como tambin la imposibilidad de que tengamos la propiedad de $C^1$ dado que trabajamos con una topologa dbil.}
%{\color{red} Definicin y propiedades que usaremos sobre los espacios de Orlicz y modulares}
%{\color{blue} El resultado de Bonans-Shapiro (teo. de Danskin) no es aplicable a nuestro contexto pues requeriria que, llamando $f(d,Z)=\langle d,Z\rangle$, se satisfaga que $D_df(d,Z)=\langle \cdot,Z\rangle$ fuese continua en el par $(d,Z)$. Sin embargo, sobre los $Z$ tenemos solo la topo. debil, que no basta pa lo anterior...} {\color{red} No entiendo, tenemos la topologa fuerte en $d$ y por ende $\langle d_n, z_n\rangle \to \langle d, z\rangle$}.
% \end{remark}
 
%
%   \begin{remark}{\color{blue} La demo. del Lema anterior prueba que la derivada direccional por la derecha es igual a \eqref{mmmwnrnwr}, pues se asumio en la demo que $\tau>0$. El mismo tipo de argumento muetra que la redivada por la izquierda (i.e. cuando $\tau\to 0-$) viene dado por \eqref{mmmwnrnwr} pero con un inf en vez de un sup. Si no me equivoco esto muestra que el subdiferencial de $\lambda\mapsto v(\bar{d}+\lambda \Delta d)$ es entonces el intervalo con limites dados por el inf y el sup. } {\color{red} No cach el punto, la derivada direccional se define en general con $\tau>0$ ($\Delta d$ es arbitrario).}
% \end{remark}

\bibliographystyle{plain}
\bibliography{projjf3}
\end{document}